\documentclass[11pt,a4paper]{article}

\usepackage{graphicx}
\usepackage[english]{babel}
\usepackage{latexsym}
\usepackage{url}
\usepackage{amssymb}
\usepackage{amsmath}
\usepackage{amsthm}
\usepackage{amsfonts}
\usepackage{ifthen}

\usepackage[left=1in,top=1.1in,right=1in,bottom=1.1in,nohead]{geometry}

\newtheorem{theorem}{Theorem}[section]
\newtheorem{corollary}[theorem]{Corollary}
\newtheorem{lemma}[theorem]{Lemma}

\newcommand\D{\ensuremath{\mathcal{D}}}
\newcommand{\RR}{\ensuremath{\mathbb R}}  
\newcommand\CR{\hbox{\tt cr}}		  
\newcommand{\CRA}{\hbox{\tt acr}}		  

\newcommand{\enf}{{\rm enf}}
\newcommand{\unw}{{\rm unw}}

\newcommand{\positive}{{\rm pos}}
\newcommand{\negative}{{\rm neg}}

\def\DEF#1{\textbf{\emph{#1}}}

\begin{document}

\title{Adding one edge to planar graphs makes crossing number\\
		and 1-planarity hard\thanks{A preliminary version of this work
		was presented at the 26th Annual Symposium on Computational Geometry~\cite{cm-socg-2010}.}}

\author{Sergio Cabello\thanks{Supported in part by 
        the Slovenian Research Agency, program P1-0297, projects J1-7218 and J1-4106,
		and within the EUROCORES Programme EUROGIGA (project GReGAS) of the European Science Foundation.}\\[2mm]
	Department of Mathematics\\
	FMF, University of Ljubljana\\
	Slovenia\\
	email: {\tt sergio.cabello@fmf.uni-lj.si}
\and
        Bojan Mohar\thanks{Supported in part by the ARRS, 
        Research Program P1-0297, by an NSERC Discovery Grant,
        and by the Canada Research Chair Program.}~\thanks{On leave 
        from IMFM \& FMF, Department of Mathematics,
        University of Ljubljana,
        1000 Ljubljana, Slovenia.}\\[2mm]
        Department of Mathematics \\
        Simon Fraser University \\
        Burnaby, B.C. \ V5A 1S6 \\
        email: {\tt mohar@sfu.ca}}

\date{\today}

\maketitle

\begin{abstract}
A graph is \emph{near-planar} if it can be obtained from a planar graph by adding an edge.
We show the surprising fact that it is NP-hard to compute the crossing number of near-planar graphs.
A graph is \emph{1-planar} if it has a drawing where every edge is crossed by at most one other edge.
We show that it is NP-hard to decide whether a given near-planar graph is 1-planar.
The main idea in both reductions is to consider the problem of simultaneously drawing
two planar graphs inside a disk, with some of its vertices fixed at the boundary of the disk. 
This leads to the concept of anchored embedding, which is of independent interest.
As an interesting consequence we obtain a new, geometric 
proof of NP-completeness of the crossing number problem,
even when restricted to cubic graphs. This resolves 
a question of Hlin\v{e}n\'y.
\end{abstract}

\section{Introduction}

A \DEF{drawing} of a graph $G$ in the plane is a representation of $G$ 
where vertices are represented by distinct points of $\RR^2$, edges are represented 
by simple polygonal arcs in $\RR^2$ joining points that correspond to their endvertices,
and the interior of every arc representing an edge contains no points representing the vertices of $G$. 
A \DEF{crossing} of a drawing $\D$ is 
a pair $(\{e,e'\},p)$, where $e$ and $e'$ are distinct edges and $p\in \RR^2$ is 
a point that belongs to the interiors of both arcs representing $e$ and $e'$ in the drawing $\D$.
The number of crossings of a drawing $\D$ is denoted by $\CR(\D)$ and is called the
\DEF{crossing number} of the drawing.
The \DEF{crossing number} $\CR(G)$ of a graph $G$ is the minimum $\CR(\D)$
taken over all drawings $\D$ of $G$. 
A \DEF{planar graph} is a graph whose crossing number is $0$. 
A drawing $\D$ with $\CR(\D)=0$ is called an \DEF{embedding} of $G$ (in the plane).
A drawing $\D$ is a \DEF{$1$-drawing} if each edge participates in at most $1$ crossing.
A \DEF{$1$-planar graph} is a graph that has some $1$-drawing.

A graph is \DEF{near-planar} if it can be obtained from a planar graph $G$ by adding an extra edge $xy$
between vertices $x$ and $y$ of $G$. We denote such near-planar graph by $G+xy$.
(The term \emph{almost planar} has also been used for the same concept~\cite{HS,M}.)
Near-planarity is a very weak relaxation of planarity, and hence it is natural
to study properties of near-planar graphs. Graphs embeddable in the torus and
apex graphs are superfamilies of near-planar graphs.

We show that it is NP-hard to compute the crossing number of near-planar graphs. 
We also show that it is NP-hard to decide whether a given near-planar graph
is $1$-planar, even when the graph has bounded degree. 
These results are not only surprising but also fundamental. They 
provide evidence that computing crossing numbers is an extremely challenging task,
even for the simplest families of non-planar graphs.

In the course of developing our NP-hardness reductions, we introduce a new notion of
anchored drawings and anchored embeddings, whose study is of independent interest.
We also prove various related hardness results for rectilinear crossing number, anchored
crossing number, and crossing number with rotations.

We show that these problems are NP-hard using a reduction from satisfiability ({\sc SAT}).
Our reductions are based on considering drawings of two planar graphs inside 
a disk with some of its vertices at prescribed positions of the boundary. 
The reductions are inspired by the work of Werner~\cite{W}, although the details in our proofs
are essentially different.
We can then use a technique from \cite{M} to relate these drawings to drawings
of near-planar graphs.

Our approach is geometric, and in particular
we obtain a new, geometric proof of NP-completeness of the crossing number problem,
even when restricted to cubic graphs. Hardness of the crossing
number problem for cubic graphs was established by
Hlin\v{e}n\'y \cite{Hli}, who asked if one can prove this result 
by a reduction from an NP-complete geometric problem instead 
of the {\sc Linear Arrangement} problem used in his proof.

\paragraph*{Related work.}
It has been known for quite some time that it is NP-hard to compute crossing numbers of graphs.
Previous proofs involved reductions from the problem {\sc Linear Arrangement}~\cite{GJ,Hli,PSS}. 
The spirit of our reduction is completely different from previous proofs
and hence of interest in its own right. 
In particular, we provide an alternative proof that computing crossing numbers is NP-hard (even when restricted to cubic graphs).
Our NP-hardness proof is more complicated, 
but it provides the additional bonus of having control over the structure
of the graph and henceforth working for near-planar graphs.

The study of crossing numbers for near-planar graphs was initiated by Riskin~\cite{Ri},
who showed that if $G$ is a planar 3-connected cubic graph, then the 
crossing number of $G+xy$ is equal to the length of a shortest path in the geometric
dual graph of the planar subgraph $G-x-y$. 
A consequence of his result it that the crossing number of a 3-connected cubic near-planar graph can
be computed in polynomial time.
Riskin asked if a similar result holds in more general situations. 
This was disproved by Mohar~\cite{M} and Gutwenger, Mutzel, and Weiskircher~\cite{GMW}.
In fact, the result cannot be extended even assuming 5-connectivity.

For near-planar graphs of maximum degree $\Delta$, 
Hlin\v{e}n\'y and Salazar~\cite{HS} provided a $\Delta$-approximation algorithm for the crossing number.
Later, we~\cite{CM} improved the approximation factor of this algorithm 
to $\lfloor\Delta/2\rfloor$ 
using combinatorial bounds that relate the crossing number
of $G+xy$ to the number of vertex-disjoint and edge-disjoint cycles in $G$ that separate $x$ and $y$.
This separation has to be defined in a certain strong sense over all planar embeddings of $G$.
Approximation algorithms for the crossing number have been provided
for some superfamilies of near-planar graphs~\cite{CH,CHM-12,HS2}. However, it should
be noted that it was not known if computing the crossing number in 
any of those families is NP-hard. 
Combinatorial bounds have also been studied in~\cite{BPT,DV}.

Our previous paper~\cite{CM} contained a closely related result: 
namely, we showed that computing the crossing number of near-planar graphs is NP-hard
for \emph{weighted} graphs. Unfortunately, our reduction was from {\sc Partition},
and hence required weights that are not polynomially bounded in the size of the graph.
Moreover, the planarizing edge $xy$ needed large weight, so $G+xy$ could not
be transformed into an unweighted near-planar graph.
See the discussion below.
In this paper we use completely different techniques.

Kawarabayashi and Reed~\cite{KR}, improving upon a result of Grohe~\cite{G},
have shown that for each constant $k_0$ there is a linear-time algorithm 
that decides if the crossing number of an input graph is at most $k_0$. 
Hence, it is clear that in our reduction the crossing number has to be an increasing
function in the number of vertices.
The currently best approximation to the crossing number of general graphs is 
by Chuzhoy~\cite{chuzhoy}.

The concept of $1$-planar graphs was introduced by Ringel~\cite{ringel}.
The concept of $1$-planarity is more subtle and because of this, the results are scarcer.
There is a lack of fundamental basic tools needed to tackle problems about 1-planarity.
Some of them are provided by Korzhik and Mohar in~\cite{km-1-planar-appear} (preliminary version in~\cite{km-1-planar-08}), 
where they studied minimal non-1-planar graphs
and, in particular, showed that recognizing $1$-planar graphs is NP-complete. 
Our proof that recognizing $1$-planar graphs is hard even for near-planar graphs
is completely different from their proof and is also much more transparent.

\paragraph*{Weighted vs.\ unweighted edges.}
Our discussion for crossing numbers will be simplified by using weighted edges.
When each edge $e$ of $G$ has a weight $w_e\in \mathbb{N}$, the crossing number 
of a drawing $\D$ is defined as $\sum w_{e}\cdot w_{e'}$, the sum taken over all 
crossings $(\{e,e'\},p)$ in $\D$. 
The crossing number of $G$ is then defined again as the minimum $\CR(\D)$
taken over all drawings $\D$ of $G$.

Let $G$ be a weighted graph.
Consider the unweighted graph $H_G$ with $V(H_G)=V(G)$, in which there are
$w_{uv}$ subdivided ``parallel" edges between $u$ and $v$ in $H_G$, for each edge $uv\in E(G)$.
It is easy to see that $\CR(H_G)=\CR(G)$. 
If the weights of $G$ are polynomially bounded in $|V(G)|$, then
$H_G$ can be constructed in polynomial time. 
Hence, in our reduction it will be enough to describe a weighted planar graph $G$ whose weights
are polynomially bounded, and then describe which extra edge $xy$ we add. The resulting
unweighted near-planar graph is $H_G+xy$.
The additional edge $xy$ that we add must have
unit weight, as otherwise the resulting graph $H_{G+xy}$ would not be near-planar.

\paragraph*{Anchored graphs.} 
The main idea in our proof is considering a concept of anchored graphs,
and studying the crossing number and $1$-planarity of such objects. 
Although this seems to be a fundamental notion, we are not aware of any previous work
considering anchored graphs.

An \DEF{anchored graph} is a triple $(G,A_G,\pi_G)$, where $G$ is a graph, $A_G$ is a subset
of vertices of $G$, and $\pi_G$ is a cyclic ordering of $A_G$. 
For reasons that will become evident soon, we call the vertices $A_G$ \DEF{anchors}.
With a slight abuse of notation, we will sometimes use $G$ to denote an anchored graph when
the anchor set $A_G$ and the ordering $\pi_G$ are implicit.

Let $\Omega\subseteq \mathbb{R}^2$ be a topological disk whose boundary is a closed polygonal line. 
An \DEF{anchored drawing} of an anchored graph $(G,A_G,\pi_G)$
is a drawing of $G$ in $\Omega$ such that the vertices of $A_G$ are represented by points on the boundary
of the disk $\Omega$, and the cyclic ordering of the anchors $A_G$ along the boundary of $\Omega$ is $\pi_G$.
An anchored drawing without crossings is an \DEF{anchored embedding}.
An \DEF{anchored planar graph} is an anchored graph that has an anchored embedding.
The \DEF{anchored crossing number} $\CRA(G,A_G,\pi_G)$, or simply $\CRA(G)$, 
of an anchored graph $G$ is the minimum number of crossings over all anchored drawings of $G$.
An \DEF{anchored $1$-drawing} is an anchored drawing where each edge participates in at most one crossing.
An \DEF{anchored $1$-planar graph} is an anchored graph that has an anchored $1$-drawing. 

Let $(G,A_G,\pi_G)$ be an anchored graph. Any subgraph $H$ of $G$ naturally defines 
the \DEF{anchored subgraph} $(H,A_H,\pi_H)$, where $A_H=A_G\cap V(H)$
and $\pi_H$ is the restriction of $\pi_G$ to the vertices in $A_H$.
We say that an anchored graph $(G,A_G,\pi_G)$ can be \DEF{decomposed into two 
anchored graphs} if there are two anchored subgraphs $(R,A_R,\pi_R)$ and 
$(B,A_B,\pi_B)$ such that $R\cup B = G$ and $E(R)\cap E(B)=\emptyset$.
The decomposition is \DEF{vertex-disjoint} when $V(R)\cap V(B)=\emptyset$.
For helping with the exposition we will refer to $R$ as ``red" graph and to $B$ as ``blue" graph.

We will use the notation $[m]=\{0,1,\dots,m\}$.

\section{Crossing number}

In this section we consider the crossing number of anchored graphs and near-planar graphs.
In Section~\ref{sec:hard1} we show that computing the crossing number of anchored graphs is NP-hard.
In Section~\ref{sec:hard2} we show that computing the crossing number of near-planar graphs is NP-hard.
In Section~\ref{sec:extensions} we provide some extensions implied by our construction.

\subsection{Crossing number of anchored graphs}
\label{sec:hard1}

The problem of minimizing the number of crossings in drawings of anchored graphs
is of independent interest.
In this section we show that computing the crossing number of anchored graphs is hard
even in a very special case when the anchored graph is decomposed
into two vertex-disjoint planar anchored subgraphs.

\begin{theorem}
\label{thm:hard1} 
	Computing the anchored crossing number of anchored graphs 
	is NP-hard, even if the input graph is decomposed
	into two vertex-disjoint planar anchored subgraphs
	(and the decomposition is part of the input).
\end{theorem}

The rest of this section is devoted the proof of Theorem~\ref{thm:hard1}.
The reduction will be from the decision problem of satisfiability:
\begin{quote}
	{\sc SAT}.\\
	\emph{Input:} A set of $n$ variables $x_1,\dots, x_n$ and a set of $m$ disjunctive
		clauses $C_1,\dots, C_m$.
	\emph{Output:} Can we assign boolean values $T/F$ to the variables such that
		the formula $C_1\wedge \dots\wedge C_m$ is satisfied?
\end{quote}

Consider an instance $I$ to {\sc SAT}.
Henceforth, we will use $n$ to denote the number of variables 
and use $m$ to denote the number of clauses.
Let $w= 30nm$. 
It is convenient to think of $w$ as a sufficiently large weight to make the reduction work.
Let $k=(6nm+6n+2m+1) w^3 - m(w^2+w-1)$. 
The upper bound 
\begin{equation*}
	k< (6nm+6n+2m+1) w^3 \le  15nm w^3 < (w^2-1)^2
\end{equation*}
will be useful in our discussion.

\paragraph{Overview.}
We next provide an overview of our reduction. 
Our aim at this point is to provide intuition.
An example showing how the whole reduction works is given in Figure~\ref{fig:example-reduction}.
It may help getting the global picture through the discussion.
We will describe an anchored \emph{blue} planar graph $B=B(I)$ 
and an anchored \emph{red} planar graph $R=R(I)$. The graphs $B$ and $R$ will
be vertex-disjoint. 
We will then construct an anchored graph $G=G(I)$ that has a decomposition into $R$
and $B$; the graph $G$ is determined by $R,B$ and by specifying the circular ordering
of the anchors $A_B\cup A_R$. 
The weights of the edges are controlled by a parameter $w=w(n,m)$.
It will turn out from the construction that, for a certain value $k=k(n,m)$,
the anchored crossing number of $G$ is at least $k$, and that it is equal to $k$
if and only if the instance $I$ can be satisfied.

The blue graph $B$ has a grid-like structure. In an optimal drawing there
are no blue-blue crossings. The weights of the blue edges
are used to encode the clauses of the instance $I$. 
The red graph has the following structure.
For each variable $x_i$, there is a pair of `vertical paths' in the red graph;
they connect anchor $r(x_i)$ to anchor $r'(x_i)$ in Figure~\ref{fig:example-reduction}. 
The construction will enforce that in an optimal drawing such a pair will be drawn 
either to the left or to the right of the middle line, 
that is, in the lighter shaded or the darker shaded region shown in the left part of Figure~\ref{fig:columnsrows}. 
Each such option corresponds to an assignment of the variable $x_i$ as $T$ or $F$. 
For each clause $C_j$, there is a `horizontal path';
it connects anchor $r_{(0,j)}$ to anchor $r_{(2n+1,j)}$ in Figure~\ref{fig:example-reduction}. 
Such path must cross a `horizontal line' of the blue grid once.
The number of crossings with such horizontal path depends on where it crosses the `blue line',
and tells if the clause is satisfied with the assignment of the variables or not.

\begin{figure*}
	\centering
	\includegraphics[width=.98\textwidth]{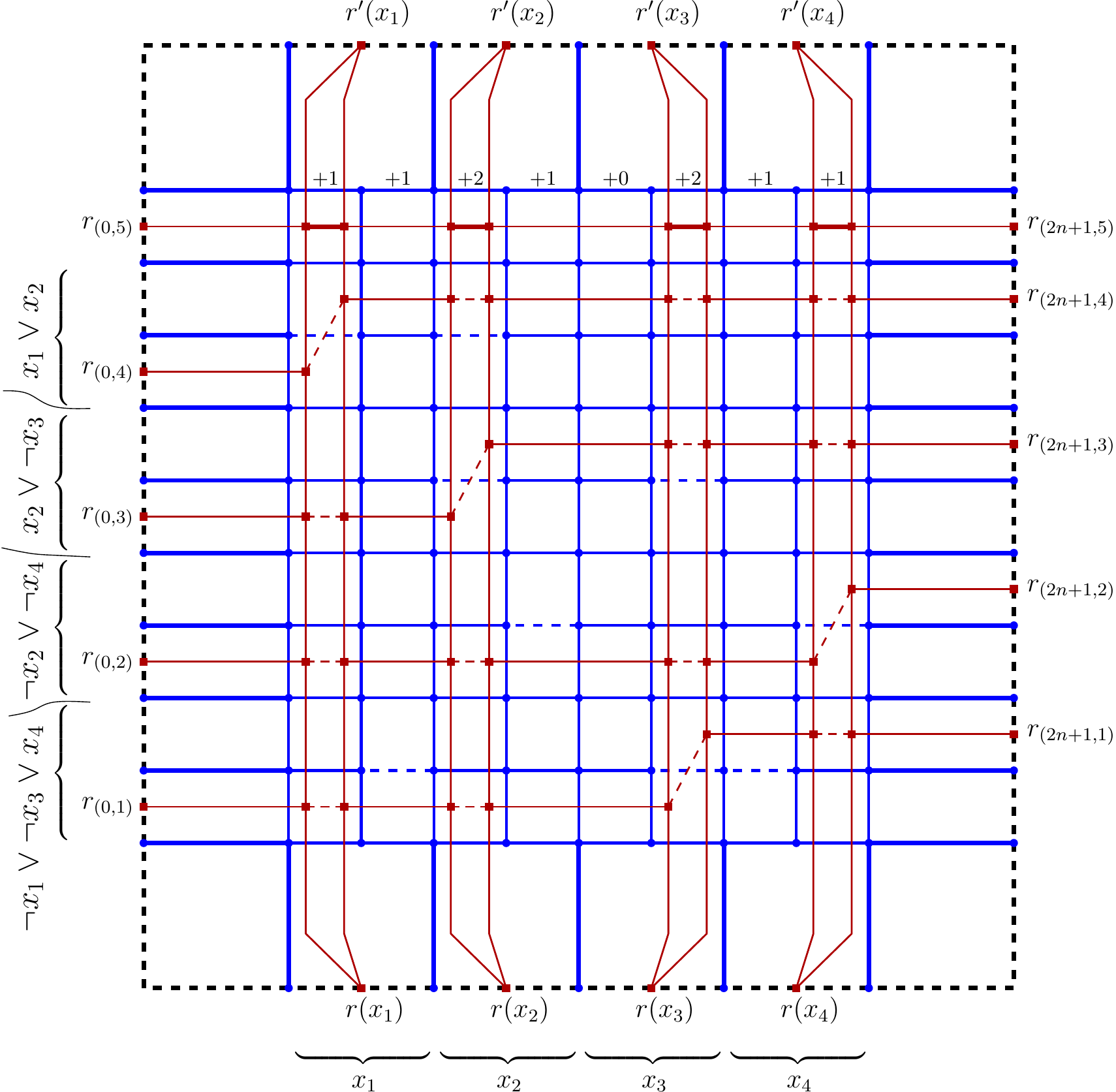}
	\caption{Example of the resulting reduction for the formula on 4 variables $x_1,x_2,x_3,x_4$
		and clauses $\neg x_1\vee \neg x_3\vee x_4$, $\neg x_2 \vee \neg x_4$,
		$x_2\vee \neg x_3$, $x_1\vee x_2$.
		The optimal drawing in the figure corresponds to the boolean assignment $x_1=x_2=T$ and $x_3=x_4=F$.
		The thickest (red or blue) edges have weight $w^4$; 
		the (blue) edges of middle thickness without annotation 
		have weight $w^2$ if solid and $w^2-1$ if dashed;
		each solid (blue) edge of middle thickness with annotation $+t$ has weight $w^2+t$;
		the thinnest (red) edges have weight $w$ if solid and $w-1$ if dashed.}
		\label{fig:example-reduction}
\end{figure*} 

\paragraph{Formal proof.}
We next proceed with the formal proof.
The blue graph $B=B(I)$ is constructed as follows (see Figure~\ref{fig:notation-blue}):
\begin{itemize}
\item[(i)] Take a grid-like graph with vertices 
	$b_{(\alpha,\beta)}$, $(\alpha,\beta)\in [2n+2]\times [2m+3]$,
	and an edge between vertices $b_{(\alpha,\beta)}$ and $b_{(\alpha',\beta')}$ 
	if and only if $|\alpha-\alpha'|+|\beta-\beta'|=1$.
\item[(ii)] Remove the vertices $b_{(2i,0)}$ and $b_{(2i,2m+3)}$ for each $i\in [n+1]$.
\item[(iii)] Define as anchors the vertices $b_{(2i+1,0)}$ and $b_{(2i+1,2m+3)}$ for each $i\in [n]$,
	and the vertices $b_{(0,\beta)}$ and $b_{(2n+2,\beta)}$ for each $\beta\in [2m+2]\setminus \{ 0\}$.
\item[(iv)] Remove the edges between any two anchors.
\item[(v)] The weights of the edges are defined as follows:
	\begin{itemize}
		\item Each edge adjacent to an anchor has weight $w^4$; 
			in Figure~\ref{fig:example-reduction} these edges are thicker.
		\item If the literal $x_i$ appears in clause $C_j$, then the edge $b_{(2i-1,2j)}b_{(2i,2j)}$ has weight $w^2-1$;
			in Figure~\ref{fig:example-reduction} these edges are dashed.
		\item If the literal $\neg x_i$ appears in clause $C_j$, then the edge $b_{(2i,2j)}b_{(2i+1,2j)}$ 
				has weight $w^2-1$; in Figure~\ref{fig:example-reduction} these edges are dashed.
		\item The edge $b_{(2i-1,2m+2)}b_{(2i,2m+2)}$ has weight $w^2+|\{ j\mid \mbox{literal $x_i$ appears in }C_j\}|$;
			in Figure~\ref{fig:example-reduction} these edges are annotated.
		\item The edge $b_{(2i,2m+2)}b_{(2i+1,2m+2)}$ has weight 
				$w^2+|\{ j\mid \mbox{literal $\neg x_i$ appears in }C_j\}|$;
			in Figure~\ref{fig:example-reduction} these edges are annotated.
		\item all other edges have weight $w^2$. 
	\end{itemize}
\end{itemize}

\begin{figure}
	\centering
	\includegraphics[width=.95\textwidth]{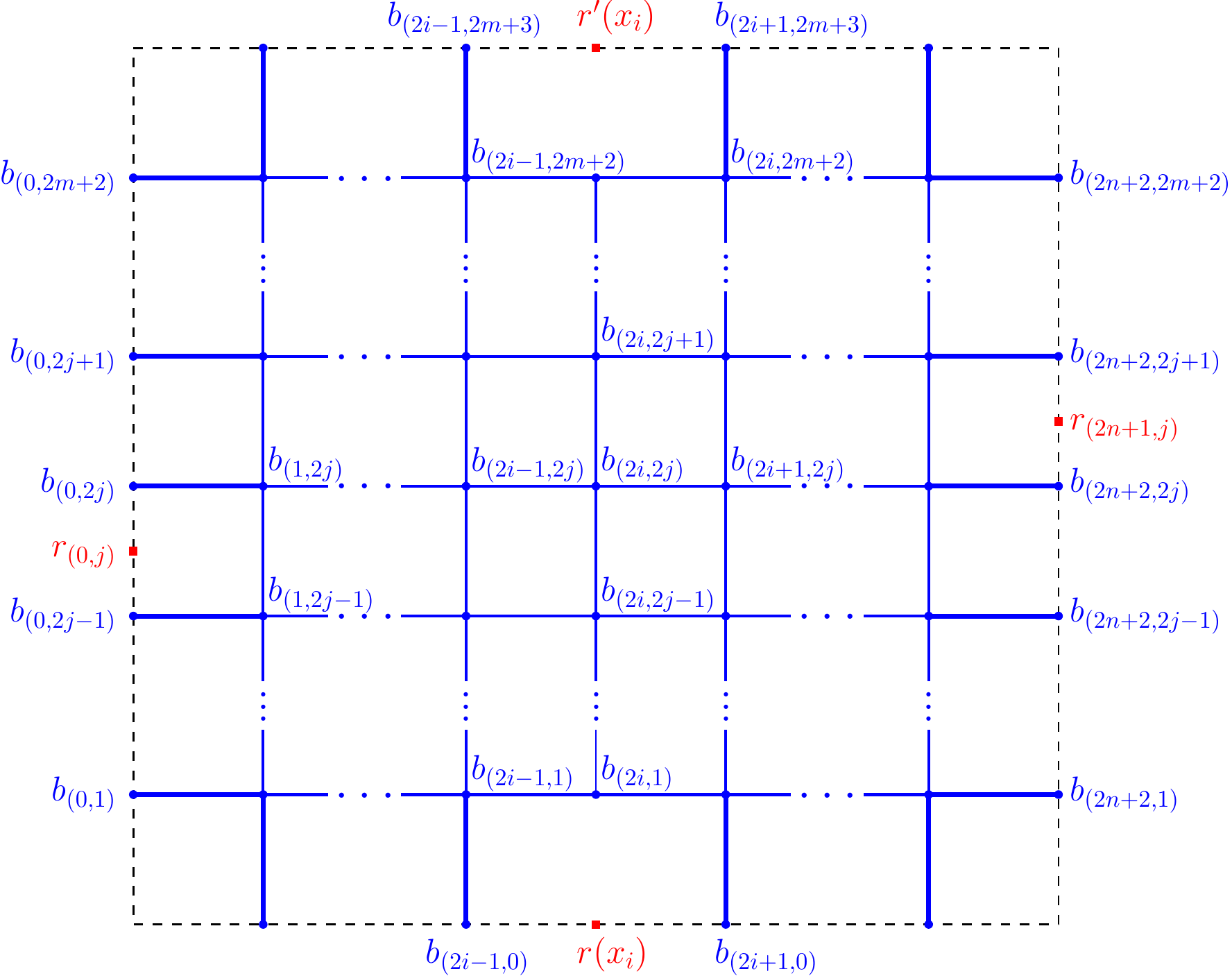}
	\caption{The graph $B=B(I)$. Anchors of the red graph are included to show the cyclic ordering of $A_B\cup A_R$.
		The thick edges have weight $w^4$. The other edges have weights between $w^2$ and $w^2-m$ depending on the instance $I$.}
		\label{fig:notation-blue}
\end{figure}

Note that each edge in the blue graph $B$ has weight at least $w^2-1$.
Hence, independently of the red graph $R$ to be defined below,
a drawing of $B$ with crossing number at most $k< (w^2-1)^2$ has to be an embedding of $B$.
Henceforth, we will assume that $B$ is anchored embedded. 
Note that the graph $B$ has a unique combinatorial embedding with anchors because of 3-connectivity.
For each variable $x_i$, we define two \DEF{columns}; see Figure~\ref{fig:columnsrows}.
The column $C_i^T$ is the region of the disk enclosed between the paths
\begin{align*}
	& b_{(2i-1,0)}b_{(2i-1,1)}\dots b_{(2i-1,2m+3)} \quad \mbox{and}  \\
	&b_{(2i+1,0)}b_{(2i+1,1)}b_{(2i,1)}b_{(2i,2)}\dots b_{(2i,2m+2)}b_{(2i+1,2m+2)}b_{(2i+1,2m+3)},
\end{align*}
and the column $C_i^F$ is the region of the disk enclosed between the paths
\begin{align*}
	& b_{(2i-1,0)}b_{(2i-1,1)}b_{(2i,1)}b_{(2i,2)}\dots b_{(2i,2m+2)}b_{(2i-1,2m+2)}b_{(2i-1,2m+3)} \quad \mbox{and}\\
	& b_{(2i+1,0)}b_{(2i+1,1)}\dots b_{(2i+1,2m+3)}.
\end{align*}

The blue edges of the form $b_{(\alpha,\beta)}b_{(\alpha+1,\beta)}$ are called \DEF{horizontal}.
The blue edges of the form $b_{(\alpha,\beta)}b_{(\alpha,\beta+1)}$ are called \DEF{vertical}.
The weights of the horizontal edges $b_{(2i-1,\beta)}b_{(2i,\beta)}$
contained in the column $C_i^T$ have been chosen
so that they add up to $2(m+1)w^2$: each time we have a $-1$ in the weight of $b_{(2i-1,2j)}b_{(2i,2j)}$
we have a $+1$ in the weight of $b_{(2i-1,2m+2)}b_{(2i,2m+2)}$. 
A similar statement holds for the column $C_i^F$: 
the weights of the horizontal edges $b_{(2i,\beta)}b_{(2i+1,\beta)}$
contained in the column $C_i^F$ add up to $2(m+1)w^2$.

For each clause $C_j$, we define two \DEF{rows}; see Figure~\ref{fig:columnsrows}.
The \DEF{upper row} $U_j$ is the region of the disk enclosed between the paths
$$
   b_{(0,2j)}b_{(1,2j)}\dots b_{(2n+2,2j)} \quad \mbox{and} \quad 
   b_{(0,2j+1)}b_{(1,2j+1)}\dots b_{(2n+2,2j+ 1)},
$$
and the \DEF{lower row} $L_j$ is the region of the disk enclosed between the paths
$$
   b_{(0,2j-1)}b_{(1,2j-1)}\dots b_{(2n+2,2j-1)} \quad \mbox{and} \quad
   b_{(0,2j)}b_{(1,2j)}\dots b_{(2n+2,2j)}.
$$

There is an additional row, called \DEF{enforcing row} and denoted by $R_\enf$, which is the
region of the disk enclosed between the paths
$$
	b_{(0,2m+1)}b_{(1,2m+1)}\dots b_{(2n+2,2m+1)}\quad \mbox{and} \quad
	b_{(0,2m+2)}b_{(1,2m+2)}\dots b_{(2n+2,2m+2)}.
$$
The role of the enforcing row $R_\enf$ is to reduce the number of possible drawings
to a well-structured subset of drawings.
We will use the columns and the rows as some sort of coordinate system to
tell where some red vertices go. For example, we may refer to
the face $C_i^T \cap L_j$.

\begin{figure}
	\centering
	\includegraphics[width=\textwidth]{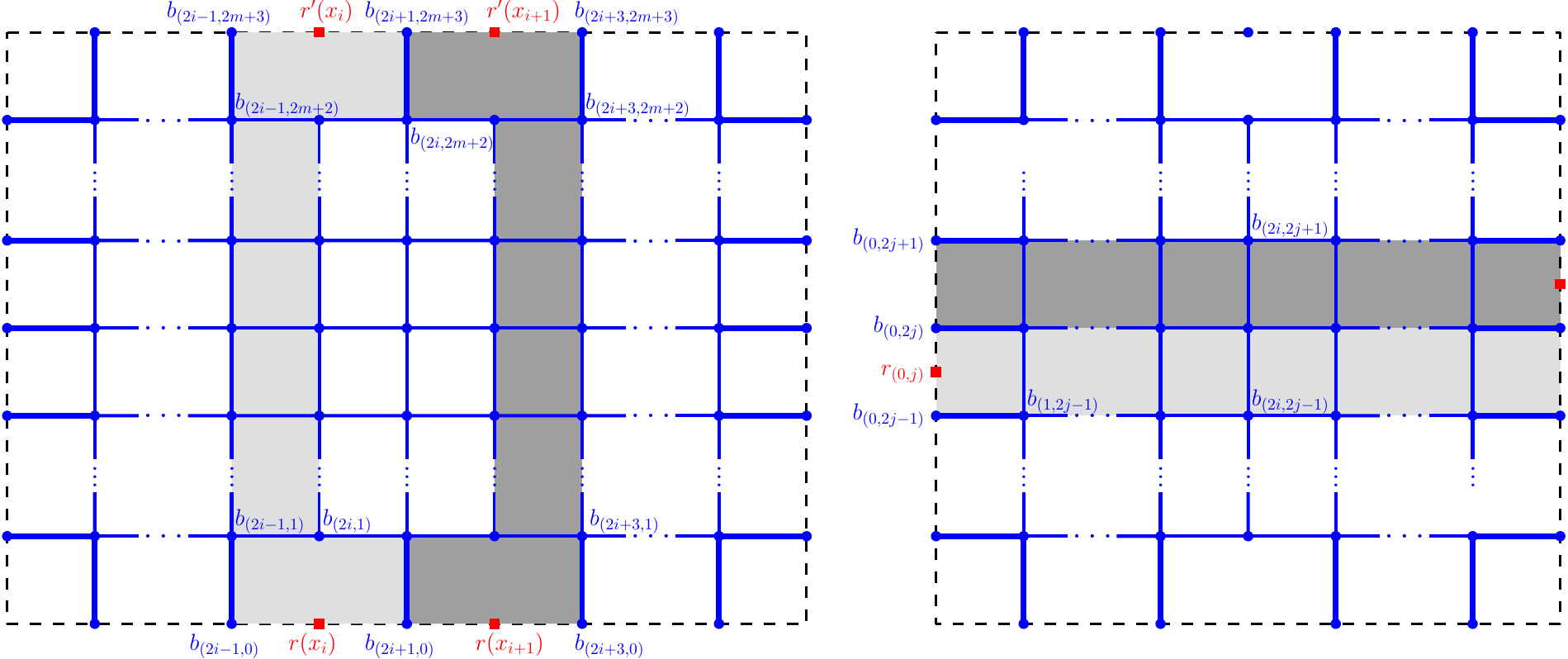}
	\caption{Left: columns $C_i^T$ (lighter shading) and $C_{i+1}^F$ (darker shading).
			Right: the upper row $U_j$ (darker shading) and the lower row $L_j$ (lighter shading).}
		\label{fig:columnsrows}
\end{figure}

The red graph $R=R(I)$ is constructed as follows (see Figure~\ref{fig:notation-red}):
\begin{itemize}
\item[(i)] Take a grid-like graph with vertices 
	$r_{(\alpha,\beta)}$, $(\alpha,\beta)\in [2n+1]\times [m+2]$,
	and an edge between vertices $r_{(\alpha,\beta)}$ and $r_{(\alpha',\beta')}$ 
	if and only if $|\alpha-\alpha'|+|\beta-\beta'|=1$.
\item[(ii)] Remove the four vertices $r_{(0,0)}$, $r_{(0,m+2)}$, $r_{(2n+1,0)}$, and $r_{(2n+1,m+2)}$.
\item[(iii)] For each variable $x_i$, 
	identify the vertices $r_{(2i-1,0)}$ and $r_{(2i,0)}$ 
	into a new vertex called $r(x_i)$, 
	and identify the vertices $r_{(2i-1,m+2)}$ and $r_{(2i,m+2)}$ 
	into a new vertex called $r'(x_i)$.
	For each variable $x_i$, the vertices $r(x_i),r'(x_i)$ are anchors for $R$.
	The vertices $r_{(0,j)}$ and $r_{(2n+1,j)}$ are also anchors for $R$, for every $j\in [m+1]\setminus \{0\}$.
\item[(iv)] Remove the edges between any two anchors.
\item[(v)] The weights of the edges are defined as follows:
	\begin{itemize}
		\item For each variable $x_i$ the edge $r_{(2i-1,m+1)}r_{(2i,m+1)}$ has weight $w^4$;
			in Figures~\ref{fig:example-reduction} and~\ref{fig:notation-red} these edges are thicker.
		\item For each variable $x_i$ and each clause $C_j$ the edge $r_{(2i-1,j)}r_{(2i,j)}$ has weight $w-1$;
			in Figures~\ref{fig:example-reduction} and~\ref{fig:notation-red} these edges are dashed.
		\item All other edges have weight $w$. 
	\end{itemize}
\end{itemize}

\begin{figure}
	\centering
	\includegraphics[width=.95\textwidth]{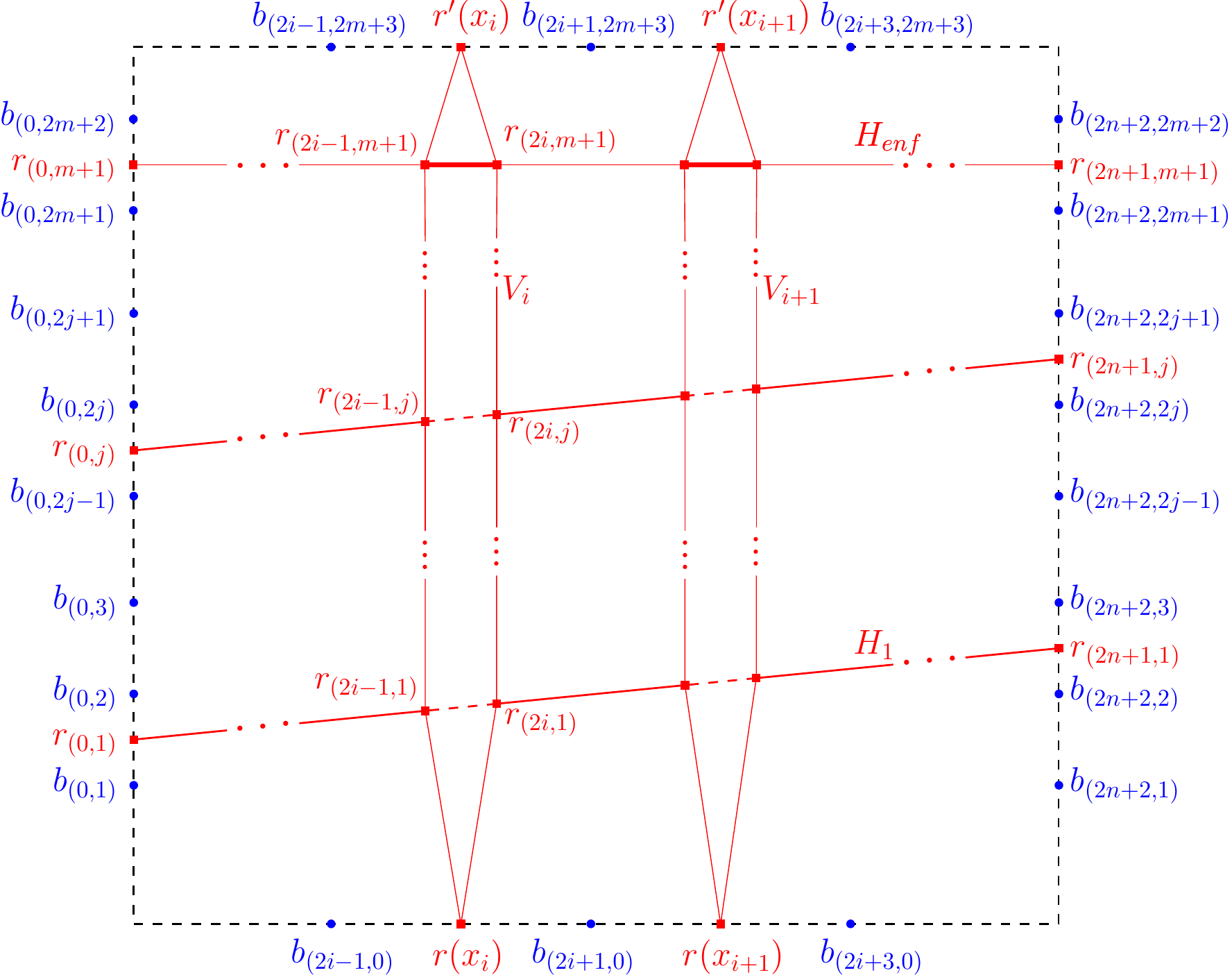}
	\caption{The graph $R=R(I)$. Anchors of the blue graph are also included to show the cyclic order of $A_B\cup A_R$.
		Thick edges have weight $w^4$; dashed edges have weight $w-1$; the remaining edges have weight $w$.}
	\label{fig:notation-red}
\end{figure}

For each variable $x_i$ the following two \DEF{vertical paths} are important: 
\begin{align*} 
	& r(x_i)r_{(2i-1,1)}r_{(2i-1,2)}\dots r_{(2i-1,m+1)} r'(x_i) \quad \mbox{and} \\
	& r(x_i)r_{(2i,1)}r_{(2i,2)}\dots r_{(2i,m+1)} r'(x_i).
\end{align*}
We will use $V_i$ to denote their union.
For each clause $C_j$, we will consider the \DEF{horizontal path} $H_j$ defined by 
\[
	r_{(0,j)} r_{(1,j)}\dots r_{(2n+1,j)}.
\]
We also define the \DEF{horizontal enforcing path} $H_\enf$ as
\[
	r_{(0,m+1)} r_{(1,m+1)}\dots r_{(2n+1,m+1)}.
\] 
The role of the horizontal enforcing path $H_\enf$ will be to reduce the number of possible drawings
to a well-structured subset of drawings.

It is important to note that the paths
$V_1,V_2,\dots ,V_n, H_1,H_2\dots, H_m, H_\enf$ form a partition of the
edge set of the red graph $R$. Hence, we can add the number of crossings that
each of them contributes separately to obtain the crossing number of a drawing.
Note also, that the intersection of a vertical pair of paths $V_i$ 
and a horizontal path $H_j$ (or $H_{\enf}$) always consists of two vertices.

Let $G=G(I)$ be the anchored graph obtained by joining the red graph $R$ and the blue graph $B$.
The (clockwise) circular ordering of the anchors along the boundary of the disk is as follows:
\begin{itemize}
	\item For each clause $C_j$ we have the sequence of anchors $b_{(0,2j-1)}$, $r_{(0,j)}$, $b_{(0,2j)}$, $b_{(0,2j+1)}$,
		and the sequence $b_{(2n+2,2j+1)}$, $r_{(2n+1,j)}$, $b_{(2n+2,2j)}$, $b_{(2n+2,2j-1)}$.
		Hence, the anchor $r_{(0,j)}$ is in the lower row $L_j$ and anchor $r_{(2n+1,j)}$ is in the upper row $U_j$.
	\item For each variable $x_i$, the anchor $r(x_i)$ is between 
		$b_{(2i-1,0)}$ and $b_{(2i+1,0)}$,
		and the anchor $r'(x_i)$ is between $b_{(2i-1,2m+3)}$ and $b_{(2i+1,2m+3)}$. 
		Hence, $r(x_i)$ and $r'(x_i)$ are in both columns $C_i^T$ and $C_i^F$.
	\item Anchor $r_{(0,m+1)}$ is between $b_{(0,2m+1)}$ and $b_{(0,2m+2)}$,
		anchor $r_{(2n+1,m+1)}$ is between $b_{(2n+2,2m+2)}$ and $b_{(2n+2,2m+1)}$.
		Hence, anchors $r_{(0,m+1)}$ and $r_{(2n+1,m+1)}$ are in the enforcing row $R_{\enf}$.
	\item Anchor $b_{(0,1)}$ comes immediately after $b_{(1,0)}$,
		anchor $b_{(1,2m+3)}$ comes immediately after $b_{(0,2m+2)}$,
		anchor $b_{(2n+2,2m+2)}$ comes immediately after $b_{(2n+1,2m+3)}$,
		and anchor $b_{(2n+1,0)}$ comes immediately after $b_{(2n+2,1)}$.		
\end{itemize}

This completes the description of the anchored graph $G$. 
We first show the easy direction of the proof, which will also give an idea of how the reduction works.
Recall that we have defined $k=(6nm+6n+2m+1)w^3 - m(w^2+w-1).$

\begin{lemma}
\label{le:easydirection}
	If the instance $I$ is satisfiable, then there is an anchored drawing of $G$ with $k$ crossings.
	Moreover, the restriction of the drawing to $R$ or to $B$ is an embedding.
\end{lemma}
\begin{proof}
	We draw the blue graph $B$ without crossings. The corresponding embedding is unique.
	The red graph $R$ is also going to be drawn without crossings. Hence, it is
	enough to describe in which face of $B$ is each red vertex and (when not obvious) where
	the red edges cross the blue edges. 
	See Figure~\ref{fig:example-reduction} for a particular example.
	Let $b_i\in \{T,F\}$ be an assignment for each variable $x_i$ of $I$ that
	satisfies all clauses. We draw the two red vertical paths of $V_i$ inside the column $C_i^{b_i}$.
	For each clause $C_j$ we proceed as follows. 
	Let $x_{t}$, where $t=t(j)$, be a variable whose value $b_{t}$ makes
	the clause $C_j$ true. We then draw the horizontal path $H_j$ as follows:
	the subpath of $H_j$ between $r_{(0,j)}$ and $r_{(2t -1, j)}$ is drawn in the lower row $L_j$,
	the edge $r_{(2t-1,j)}r_{(2t,j)}$ crosses from $L_j$ to $U_j$ through the blue edge
	in $L_j\cap U_j\cap C_{t}^{b_{t}}$,
	and the subpath of $H_j$ between $r_{(2t, j)}$ and $r_{(2n+1, j)}$ is drawn in the upper row $U_j$.
	The path $H_{\enf}$ is drawn inside the row $R_{\enf}$.
	Note that this description implicitly assigns to each non-anchor vertex of $R$ a face of $B$.
	The drawing can be extended to a planar embedding of $R$ in such a way that
	no red edge crosses twice any blue edge.	
	
	Let us now compute the crossing number of the drawing we have described.
	There are no monochromatic crossings in the construction; hence we only need to count
	the red-blue crossings.
	Each of the two paths in $V_i$ contributes $2(m+1)w^3$ to the crossing number of the drawing:
	edges in $V_i$ have weight $w$, and each path in $V_i$ crosses all the horizontal blue
	edges contained in $C_i^{b_i}$, whose weights add to $2(m+1)w^2$.
	Each horizontal path $H_j$ contributes $(2n+1)w^3+(w-1)(w^2-1)$ to the crossing number of the drawing:
	the edges on the horizontal path $H_j$ connecting $V_i$ to $V_{i+1}$ 
	have weight $w$ and cross $2n+1$ blue vertical edges whose weight is $w^2$;
	there is only one red edge in $H_j$, namely $r_{(2t(j)-1,j)}r_{(2t(j),j)}$ with weight $w-1$,
    that crosses the boundary between rows $L_i$ and $U_i$, namely at the edge
	of $L_j\cap U_j\cap C_{t}^{b_{t}}$ with weight $w^2-1$ because 
	the corresponding literal $x_i$ or $\neg x_i$ makes $C_j$ satisfied. 
	(Note that if the literal $x_i$ or $\neg x_i$ would not satisfy $C_j$, the weight
	of the crossed blue edge would be $w^2$, which would mean an increment of $cr(\D)$ by $w-1$. We will
	use this fact to prove the opposite statement of the Lemma below.)
	The horizontal path $H_{\enf}$ contributes $(2n+1)w^3$ to the crossing number of the drawing:
	the edges of $H_j$ connecting $V_i$ to $V_{i+1}$ 
	have weight $w$ and cross $2n+1$ blue edges whose weight is $w^2$.

	The crossing number of the drawing is thus
	\[
		n\cdot 2\cdot 2(m+1)w^3+ m\cdot \bigl((2n+1)w^3 + (w-1)(w^2-1)\bigr) + (2n+1)w^3 
	\]
	which is 
	\[
		(6nm+6n+2m+1) w^3 - m(w^2+w-1)\,=\, k.\qedhere
	\]
\end{proof}

We next have to show the reverse implication: if the anchored crossing number of $G$
is at most $k$, then the formula $I$ is satisfiable.
Henceforth, let us assume for the rest of this section that $\CRA(G)\le k$,
and let us fix an anchored drawing $\D$ of $G$ with at most $k$ crossings.
As mentioned before, $\D$ cannot have any blue-blue crossing because otherwise $\CR(\D)>k$.
In principle, $\D$ could contain red-red crossings; we will show below 
that in fact this is not possible, and hence all crossings are red-blue.
It will be convenient to look at the number of red-blue crossings without taking into account
the weights. We refer to such crossings as \DEF{unweighted crossings}.
Simple arithmetic shows the following two properties.

\begin{lemma}
\label{le:property1}
	The drawing $\D$ has at most $6nm+6n+2m+1$ unweighted red-blue crossings.
\end{lemma}
\begin{proof}
	Each blue edge has weight at least $w^2-1$ and each red edge has weight at least $w-1$.
	Thus each red-blue crossing contributes weight at least $(w-1)(w^2-1)$ towards the crossing number
	of $\D$. If there were strictly more than $6nm+6n+2m+1$ red-blue crossings, then the weighted crossing
	number of the drawing would be at least
	\begin{align*}
		(6nm+6n+2m+2&)(w-1)(w^2-1) \\
			&=\, (6nm+6n+2m+2)(w^3-w^2-w+1) \\
			&= \, (6nm+6n+2m+1)(w^3-w^2-w+1)+ (w^3-w^2-w+1)\\
			&= \, k - (6nm+6n+m+1)(w^2+w-1)+ (w^3-w^2-w+1)\\
			&= \, k + w^3 - (6nm+6n+m+2)(w^2+w-1)\\
			&> \, k + w^3 - 28nmw^2\\
			&> \, k.
	\end{align*}
	Hence there are at most $6nm+6n+2m+1$ red-blue unweighted crossings.
\end{proof}

Using Lemma~\ref{le:property1} and the properties of the enforcing row $R_\enf$
and the horizontal path $H_\enf$
we can show that the drawing $\D$ has the following structure.

\begin{lemma}
\label{le:property2}
	If $\CR(\D) \le k$, then:
	\begin{itemize}
		\item[\textup{(i)}] For each clause $C_j$, the horizontal path $H_j$ is inside the rows $U_j\cup L_j$
		and crosses precisely $2n+2$ blue edges.
		\item[\textup{(ii)}] The horizontal path $H_{\enf}$ is drawn inside the row $R_{\enf}$.
		\item[\textup{(iii)}] For each variable $x_i$, both vertical paths of $V_i$ are inside the column $C_i^T$
		or both are inside the column $C_i^F$.
	\end{itemize}
\end{lemma}
\begin{proof}
	As mentioned before, $\D$ cannot have any blue-blue crossing because otherwise $\CR(\D)>k$.
	Moreover, no blue edge incident to an anchor
	crosses any other edge because it has weight $w^4$.
	
	Through this proof, we use $\CR_{\unw}(X)$ to denote the number of unweighted
	red-blue crossings of a red subgraph $X$ with the blue graph.
	Each of the two red vertical paths in $V_i$ 
	crosses each of the horizontal rows $R_{\enf}$ and $U_j,L_j$.
	Therefore 
	\[
		\CR_{\unw}(V_i)\ge 2\cdot (2m+2).
	\]
	Any horizontal path $H_j$
	crosses each of the columns $C_i^T,C_i^F$ and crosses
	the boundary between rows $L_j$ and $U_j$. (Here we need that edges incident
	to the anchors do not participate in any crossing, as otherwise $H_1$ could go below $L_1$.)
	Therefore 
	\[
		\CR_{\unw}(H_j)\ge  2n+2.
	\]
	Similarly, for the horizontal path $H_{\enf}$ it holds
	\[
		\CR_{\unw}(H_{\enf})\ge  2n+1.
	\]
	We conclude that  
	\begin{align*}
		\CR_{\unw}(R) &\,= \, \sum_{i=1}^n \CR_{\unw}(V_i) + 
		   \sum_{j=1}^m \CR_{\unw}(H_j) + \CR_{\unw}(H_{\enf})\\
		& \ge\, 2n(2m+2) + m (2n+2) + 2n+1 \\
		& =\, 6nm+6n+2m+1.
	\end{align*}
	Since by Lemma~\ref{le:property1} we have 
	$\CR_{\unw}(R) \le 6nm+6n+2m+1$,
	we conclude that
	\begin{align*}
		\CR_{\unw}(V_i)&= 4m+4,\\
		\CR_{\unw}(H_j)&=  2n+2,\\
		\CR_{\unw}(H_{\enf})&=2n+1.
	\end{align*}
	Since no blue edge adjacent to an anchor can be crossed by any other edge, these
	equalities imply that each one of the paths in $V_i$ is contained in the column
	$C_i^T$ or in $C_i^F$,
	that the path $H_j$ is contained in the rows $L_j\cup U_j$,
	and that the path $H_{\enf}$ is contained in the row $R_{\enf}$.

	We next argue that both paths in $V_i$ are in $C_i^T$ or both are in $C_i^F$.
	For this, consider the edge $e_i$ of the horizontal enforcing path $H_{\enf}$ that connects the two paths of $V_i$.
	Since this edge $e_i$ has weight $w^4$ it cannot cross any other edge in the drawing. In particular
	the endpoints of $e_i$ must be in the same face, say $f$, of the embedding of the blue graph $B$.
	Since this face $f$ has to be in $R_{\enf}$, both endpoints of $e_i$ have to be in $C_i^T$ or
	$C_i^F$, and both paths from $V_i$ are in the same column.	
\end{proof}

\begin{lemma}
	\label{le:harddirection}
	If $\CR(\D) \le k$, then the instance
	$I$ is satisfiable. Moreover, the restriction of the drawing $\D$ to $R$ or to $B$ is an embedding.
\end{lemma}
\begin{proof}
	By Lemma~\ref{le:property2}, in the drawing $\D$
	both paths in $V_i$ are contained either in $C_i^T$ or $C_i^F$.
	Consider the assignment where variable $x_i$ gets value $b_i=T$ if
	the two paths of $V_i$ are contained in $C_i^T$, and $b_i=F$ otherwise.
	We will show that this assignment satisfies the formula $C_1\wedge\dots\wedge C_m$ of the instance $I$.
	
	We will use in our analysis the properties of $\D$ obtained in Lemma~\ref{le:property2}.
	For a red subgraph $X$, let $\CR(X)$ denote the crossing number of the subdrawing of $\D$ induced
	by $X$ and the blue graph $B$. 
	
	All the edges in $V_i$ have weight $w$. The weights of the horizontal blue edges
	contained in $C_i^{b_i}$ add to $2(m+1)w^2$. Furthermore, note that each of those blue horizontal
	edges is crossed by each of the two paths in $V_i$. 
	Therefore we have $\CR(V_i)=2\cdot (2m+2)w^3$,
	and thus 
	\[
		\CR(\cup_i V_i)=(4nm+4n)w^3.
	\]
	All the edges from the path $H_{\enf}$ have weight $w$ or $w^4$, and hence only
	edges from $H_\enf$ with weight $w$ may cross the vertical blue edges contained in the row $R_\enf$.
	Inside the row $R_\enf$ there are $2n+1$ vertical blue edges of weight $w^2$,
	and thus 
	\[
		\CR(H_{\enf})=(2n+1)w^3.
	\]
	
	Since $cr(\D)\le k$ and the paths $V_1,\dots,V_n, H_1,\dots H_m,H_\enf$
	form an edge-disjoint partition of the edges of $R$, we have
	\begin{equation*}
		\CR(\cup_j H_j) + \CR(\cup_i V_i) + \CR(H_\enf) \,\le\, k,
	\end{equation*}
	and therefore
	\begin{equation}\label{eq:1Hj}
		\CR(\cup_j H_j)\le k- (4nm+4n)w^3 - (2n+1)w^3 = m\bigl( (2n+2) w^3 - (w^2+w-1)\bigr).
	\end{equation}
	For each clause $C_j$, the edges of $H_j$ have weight $w$, if they connect a vertex 
	in $V_i\cap H_j$ to a vertex in $V_{i+1}\cap H_j$ for some $i$,
	or weight $w-1$ if they connect both vertices of $V_i\cap H_j$.
	Because of Lemma~\ref{le:property2}, 
	the edges with weight $w-1$ are always within the column $C_i^{b_i}$ for some $i$.
	This means that the boundary of any column $C_i^T$ or $C_i^F$, which has weight $w^2$,
	is always crossed by a red edge
	of $H_j$ with weight $w$. Let $\partial_j$ denote the boundary between the lower row $L_j$ and the upper row $U_j$.
	This boundary $\partial_j$ must also be crossed by an edge of $H_j$, 
	and that crossing contributes weight at least $(w-1)(w^2-1)$ to $\CR(H_j)$. We conclude that 
	\begin{equation}\label{eq:2Hj}
		\CR(H_j)\, \ge\, (2n+1)\cdot w\cdot w^2 + (w-1)(w^2-1) \,=\, (2n+2)w^3 - (w^2+w-1),
	\end{equation}
	with equality if and only if the crossing between $H_j$ and $\partial_j$ contributes
	exactly $(w-1)(w^2-1)$ to the crossing number.
	Combining equations (\ref{eq:1Hj}) and (\ref{eq:2Hj}), we see that 
	\begin{equation}\label{eq:3Hj}
		\CR(H_j) = (2n+2)w^3 - (w^2+w-1)
	\end{equation}
	for each clause $C_j$.
	Therefore, the boundary $\partial_j$ must be crossed at a blue edge $b_{(t,j)}b_{(t+1,j)}$ of weight
	$w^2-1$ by a red edge $r_{(2i-1,j)}r_{(2i,j)}$ of weight $w-1$. It may be that $t=2i-1$ or $t=2i$.
	Consider first the case when $t=2i-1$. By the construction of the blue graph $B$,
	the edge $b_{(t,j)}b_{(t+1,j)}$ has weight $w^2-1$ because the literal $x_i$ appears in the clause $C_j$ of $I$.
	Moreover, the endpoints of $r_{(2i-1,j)}r_{(2i,j)}$ must be in the column $C_i^T$ since they are part of the
	vertical paths $V_i$. 
	Hence, the clause $C_j$ is satisfied by the assignment $x_i=b_i=T$ we defined at the beginning of the proof.
	The case when $t=2i$ is alike, but in this case the literal $\neg x_i$ appears in the clause $C_j$ of $I$,
	and we took the assignment $x_i=b_i=F$.
	
	Finally, note that our analysis shows that in $\D$ there are \emph{exactly} $k$ red-blue crossings,
	and therefore there cannot be any red-red crossings. 
	Hence the restriction of the drawing $\D$ to the red graph 
	is an embedding.
\end{proof}

We can now prove our main result.
\begin{proof}[Proof of Theorem~\ref{thm:hard1}] 
	Given an instance $I$ for SAT, we construct anchored graphs 
	$R=R(I)$, $B=B(I)$, and $G=G(I)$ as
	described in the text. We further construct, for each $X\in \{G,B,R\}$,
	the graph $H_X$ obtained from $X$ by replacing each edge $uv\in E(X)$ 
	of weight $w_{uv}$ by $w_{uv}$ parallel paths of length $2$ connecting $u$ to $v$.
	It is clear that $H_R$ and $H_B$ is a vertex-disjoint decomposition of $H_G$ into planar anchored graphs.
	Since the weights of $X$ are bounded by a polynomial in $n$ and $m$, 
	it follows that the graphs $H_G,H_R,H_B$ can be constructed in polynomial time. 
	As discussed in the introduction, we have $\CRA(H_G)=\CRA(G)$.
	From Lemmas~\ref{le:easydirection} and~\ref{le:harddirection} if follows that
	$\CRA(H_G)=\CRA(G) \le k$ if and only if $I$ is satisfiable.
\end{proof}
	
\subsection{Crossing number for near-planar graphs}
\label{sec:hard2}

We can now show that computing the crossing number of near-planar graphs is NP-hard.
Our reduction is from the problem in Theorem~\ref{thm:hard1}, and we make use
of the fact that the anchored graph is decomposed into two vertex-disjoint planar anchored graphs $R$ and $B$.
The high-level approach is the following: we replace each edge of $R\cup B$
by an edge with heavy weight and replace the boundary of the disk by a cycle $C$ with heavy weights. 
The resulting graph is planar. We make it near-planar adding an arbitrary edge connecting
$R$ to $B$. The heavy weights of $C$ force that, in an optimal drawing, $C$ is embedded.
Moreover, the additional edge forces the graphs $R$ and $B$ 
to be drawn on the same side of $C$, thus resembling an anchored drawing.

\begin{theorem}
\label{thm:hard10} 
	Computing the crossing number of near-planar graphs is NP-hard.
\end{theorem}
\begin{proof}
	Consider an anchored graph $G$, possibly with weighted edges, that is 
	decomposed into two vertex-disjoint planar anchored graphs $R$ and $B$.
	Let $W$ be the sum of the weights of the edges in $G$; if $G$ is unweighted,
	then $W$ is the number of edges in $G$.
	The construction will use a parameter $\lambda = 320 Wnm > 2W$.
	(In this proof, setting $\lambda=2W+1$ would be enough, but
	we will need such larger $\lambda$ in the proof of Corollary~\ref{co:hard6} below.)
	
	Consider the weighted graph $G'=G'(G,\lambda)$, without anchors, 
	obtained from $G$ as follows: we start with $G$ and multiply the weight of each edge by $\lambda$. 
	For every two consecutive anchors $a$ and $a'$ of $G$ in the cyclic ordering $\pi_G$,
	we introduce in $G'$ an edge $aa'$ with weight $\lambda^4$. The set of added edges
	defines a cycle, which we denote by $C=C(G,\lambda)$. This completes the description of $G'$.
	We also fix an arbitrary vertex $r$ of $R$ that is not an anchor,
	an arbitrary vertex $b$ of $B$ that is not an anchor.
	(If all vertices of $R$ are anchors we just subdivide
	an edge and take $r$ as the new vertex. A similar procedure can be done with $B$.)
	We will study the crossing number of $G'+rb$.

	Firstly, we show that $G'$ is planar, and thus $G'+rb$ is a near-planar graph. 
	The graph $G'$ consists of
	the cycle $C$ connecting consecutive anchors of $G$, and two planar anchored graphs $R$ and $B$.
	We can thus embed $G'$ taking an embedding of $C$ in
	$\RR^2$, embedding $R$ in the disk bounded by $C$ in $\RR^2$,
	and embedding $B$ in the exterior of the disk. 
	Note that the embeddings of $R$ and $B$ exist because they are planar anchored graphs.
	Since the weight of the edge $rb$ is one, it follows that $G'+rb$ is a near-planar graph,
	even when replacing each edge of $G'$ by the corresponding
	number of subdivided parallel edges.
	
	We claim that $\CRA(G)=\lfloor \CR(G'+rb)/\lambda^2 \rfloor$. 
	Consider an optimal anchored drawing $\D$ of $G$ in a topological disk $\Omega$ with $\CR(\D)=\CRA(G)$. 
	We then extend the drawing $\D$ to obtain a drawing $\D'$ of $G'+rb$ as follows.
	Pushing the interior of the edges inside $\Omega$, 
	we may assume that $\D$ touches the boundary $\partial \Omega$ of $\Omega$
	only at the anchors. We then draw the edges $aa'$ of $G'$ between consecutive anchors 
	along the boundary of the disk $\Omega$. Finally, we draw the edge $rb$
	so as to minimize the number of crossings it contributes. Each crossing of $\D$
	contributes $\lambda^2$ crossings to $\D'$. The edge $rb$ can cross each edge of $G'-E(C)$ at most once 
	in the drawing $\D'$ because of optimality of $\D$ and the drawing of $rb$. 
	Therefore $\CR(\D')\le \lambda^2\cdot \CR(\D) + \lambda W$ because the sum of the edge weights in $G$ was $W$.
	Using that $\lambda > 2W$ we get 
	\begin{equation}\label{eq1}
		\CR(G'+rb)/\lambda^2 \le  \CRA(G) + W/\lambda \le \CRA(G)+ 1/2.
	\end{equation}
	
	Let $\D'$ be a drawing of $G'+rb$ in the plane such that no edge crosses
	itself and no two edges cross more than once.
	If the cycle $C$ is embedded and no edge of $C$ is involved in
	a crossing, then
	each pair of edges not in $C$ can cross at most once, and hence $\CR(D')$ is
	upper bounded by $\binom W2 \lambda^2 + 1\cdot W\cdot \lambda <\lambda^4$, where we have used again $\lambda > 2W$.
	If the restriction of $\D'$ to $C$ is not an embedding or some edge of $C$
	participates in a crossing, then $\CR(D')\ge \lambda^4$.	
	It follows that in every optimal drawing $\D'$ of $G'+rb$, the cycle $C$
	is embedded and no edge crosses $C$. 
	
	Consider now an optimal drawing $\D'$ of $G'+rb$ with $\CR(\D')=\CR(G'+rb)$.
	Let $\Omega$ be the disk bounded by the image of $C$ in $\D'$.
	Since no edge of $G'+rb$ can cross an edge of $C$, the drawing $\D'$ is contained
	in the closure of $\Omega$.
	If in $\D'$ we remove the image of $C$ and the image of the edge $rb$,
	then we obtain an anchored drawing of $G$, which we shall denote by $\D$. 
	Note that $\CR(\D)$ is equal to $\CR(\D')$ minus the number of crossings
	contributed by $rb$ and
	scaled down by $\lambda^2$ because of the weights introduced in $G'$. 
	We thus have 
	\[
		\CRA(G)\le \CRA(\D) \le \CR(\D')/\lambda^2=\CR(G'+rb)/\lambda^2.
	\]
	Combining with equation (\ref{eq1}) we get
	\begin{equation}\label{eq2}
			\CRA(G) \le \CR(G'+rb)/\lambda^2 \le  \CRA(G)+ 1/2.
	\end{equation}
	Since $\CRA(G)$ is an integer, this 
	finishes the \emph{proof of the claim $\CRA(G)=\lfloor \CR(G'+rb)/\lambda^2 \rfloor$}.
	
	The graph $G'+rb$ can be constructed from $G$ in polynomial time. 
	Moreover, since 
	since the weights of $G'$ are polynomially bounded, we can also replace
	each edge by parallel subdivided edges to obtain an unweighted graph $H_{G'}+rb$,
	as described in the introduction, that satisfies $\CR(G'+rb)=\CR(H_{G'}+rb)$.
	Since the graph $H_{G'}+rb$ is near-planar and $\CRA(G)=\lfloor \CR(H_{G'}+rb)/\lambda^2 \rfloor$,
	the result follows from Theorem~\ref{thm:hard1}.	
\end{proof}

\subsection{Extensions}
\label{sec:extensions}

Our NP-hardness proof in Section~\ref{sec:hard1}
has additional properties that imply certain extensions worth to be mentioned.
Let $G$, $R$ and $B$ be as defined in Section~\ref{sec:hard1}.
The key observation for our extensions is that
Lemmas~\ref{le:easydirection} and~\ref{le:harddirection} imply
that the following statements are equivalent:
\begin{itemize}
\item the instance $I$ is satisfiable;
\item graph $G$ has an anchored drawing with at most $k$ crossings;
\item graph $G$ has an anchored drawing with at most $k$ crossings 
	such that the restrictions to each of $R$ and $B$ are embeddings.
\end{itemize}

A direct consequence is the following.

\begin{corollary}
\label{co:hard2} 
	The following problem is NP-hard: given an anchored graph $G$ decomposed 
	into two vertex-disjoint graphs $R$ and $B$, compute the minimum of $\CRA(\D)$ 
	over the anchored drawings $\D$ of $G$ that are an embedding when restricted to $R$ and when restricted to $B$.
\end{corollary}

Since the embedding of $B$ and the embedding of $R$ are unique, we can modify the construction a bit
to obtain 3-connectivity. The approach can be summarized as follows: we multiply the weights by a large enough parameter
$\lambda$ and then add edges to increase connectivity. The contribution of the original edges
to the crossing number grows quadratically in $\lambda$ while the contribution of the new edges can be kept 
at $O(\lambda nm)$.

\begin{corollary}
\label{co:hard6} 
	The following problem is NP-hard: 
	given a simple planar, 3-connected graph $G$ and an additional edge $xy$,
	compute the crossing number of $G+xy$.
\end{corollary}
\begin{proof}
	Consider the anchored graph $G$ and the subgraphs $B$ and $R$, as constructed in Section~\ref{sec:hard1}.
	Recall the parameter $\lambda=320Wnm$, the planar graph $G'$, and the edge $rb$ used in the proof of Theorem~\ref{thm:hard10}.
		
	We now construct another planar graph $\tilde G$ from $G'$, as follows.
	See Figure~\ref{fig:replacement} for an example of the transformation locally.
	We start with a copy of $G'$ and replace each edge $uv$ of weight $w_{uv}$ with $w_{uv}$ edges, each subdivided once,
	and connect the subdividing vertices with a path. Let's call the resulting graph $G''$, which is also planar.
	We call the additional vertices \DEF{subdividing vertices}.
	Consider an embedding of $G''$, which is unique up to a permutation in each group of parallel paths. 
	Within each non-triangular face defined by blue edges and edges of $C$, 
	we add (at most 4) edges between subdividing vertices so that the subgraph of $G''$ induced by $B\cup C$ is $3$-connected.
	For the red graph $R$, we do a slightly different transformation:  within each non-triangular face with red edges 
	we add (at most 4) edges between the \emph{red subdividing vertices} so that the subgraph of $G''$ induced by $R$ is $3$-connected.
	(Thus, we do not add new edges connecting $R$ to $C$.)
	The resulting graph is $\tilde G$. 
	It is planar by construction and it is $3$-connected because $R$ has at least $3$ anchors.
	The edges added in the transformation from $G''$ into $\tilde G$ are called \DEF{additional edges}. 
	Note that the embedding of $G''$ has less than $20nm$ faces and therefore there are at most $80nm$ additional edges.
	
	The transformation from $G'$ to $G''$ is similar to the transformation from $G'$ to $H_{G'}$ discussed
	in the introduction, and thus $\CR(G''+rb)=\CR(G'+rb)$.
	From equation (\ref{eq2}) in Theorem~\ref{thm:hard10} we conclude
	\begin{equation}\label{eq3}
		\lambda^2 \CRA(G) \le \CR(G''+rb) \le \lambda^2 \CRA(G)+ \lambda^2 /2
	\end{equation}
	
	Since $G''$ is a subgraph of $\tilde G$, we have 
	\begin{equation}\label{eq4}
		\CR(G''+rb) \le \CR(\tilde G+rb)
	\end{equation}
	Consider an optimal drawing $\D''$ of $G''+rb$. 
	As discussed in the proof of Theorem~\ref{thm:hard10}, in $\D''$ the cycle $C$ is embedded
	and we may assume that $R$ and $B$ are drawn inside the disk bounded by $C$.
	We modify this drawing into a drawing of $\tilde G+rb$ in the following way.
	We redraw each edge $aa'$ in $C$ in such a way that vertices incident to any additional edge
	are in the interior of $C$. (This is the reason
	for the asymmetry treating $B$ and $R$ in the construction of $\tilde G$.
	If both $R$ and $B$ would have additional edges incident to subdividing vertices of $C$, this
	step would not be possible.) This step does not introduce any crossing.
	Then, we draw each of the $80nm$ additional edges optimally; each such edge is inside the disk bounded by $C$.
	Each additional edge can incur into at most $\lambda W+80nm$ new crossings 
	because the sum of the weights of $\tilde G-C$ is at most $\lambda W+80nm$.
	From the resulting drawing of $\tilde G+rb$ we obtain
	\[
		\CR(\tilde G+rb) \le \CR(G''+rb) + 80nm\cdot (\lambda W + 80nm) \le \CR(G''+rb) + (3/8)\lambda^2.
	\]
	Combining with equations (\ref{eq3}) and (\ref{eq4}) we get
	\[
		\lambda^2 \CRA(G)\le  \CR(\tilde G+rb) \le \CR(G''+rb) + (3/8)\lambda^2 \le \lambda^2 \CRA(G)+ (7/8)\lambda^2.
	\]
	It follows that
	\[
		\CRA(G) = \lfloor \CR(\tilde G+rb)/ \lambda^2 \rfloor,
	\]
	Since computing $\CRA(G)$ is NP-hard by Theorem~\ref{thm:hard1}, it is also NP-hard to compute $\CR(\tilde G+rb)$.
	The graph $\tilde G$ can be constructed in polynomial time, it is planar and $3$-connected,
	as desired.
\end{proof}

\begin{figure}
	\centering
	\includegraphics[width=.95\textwidth]{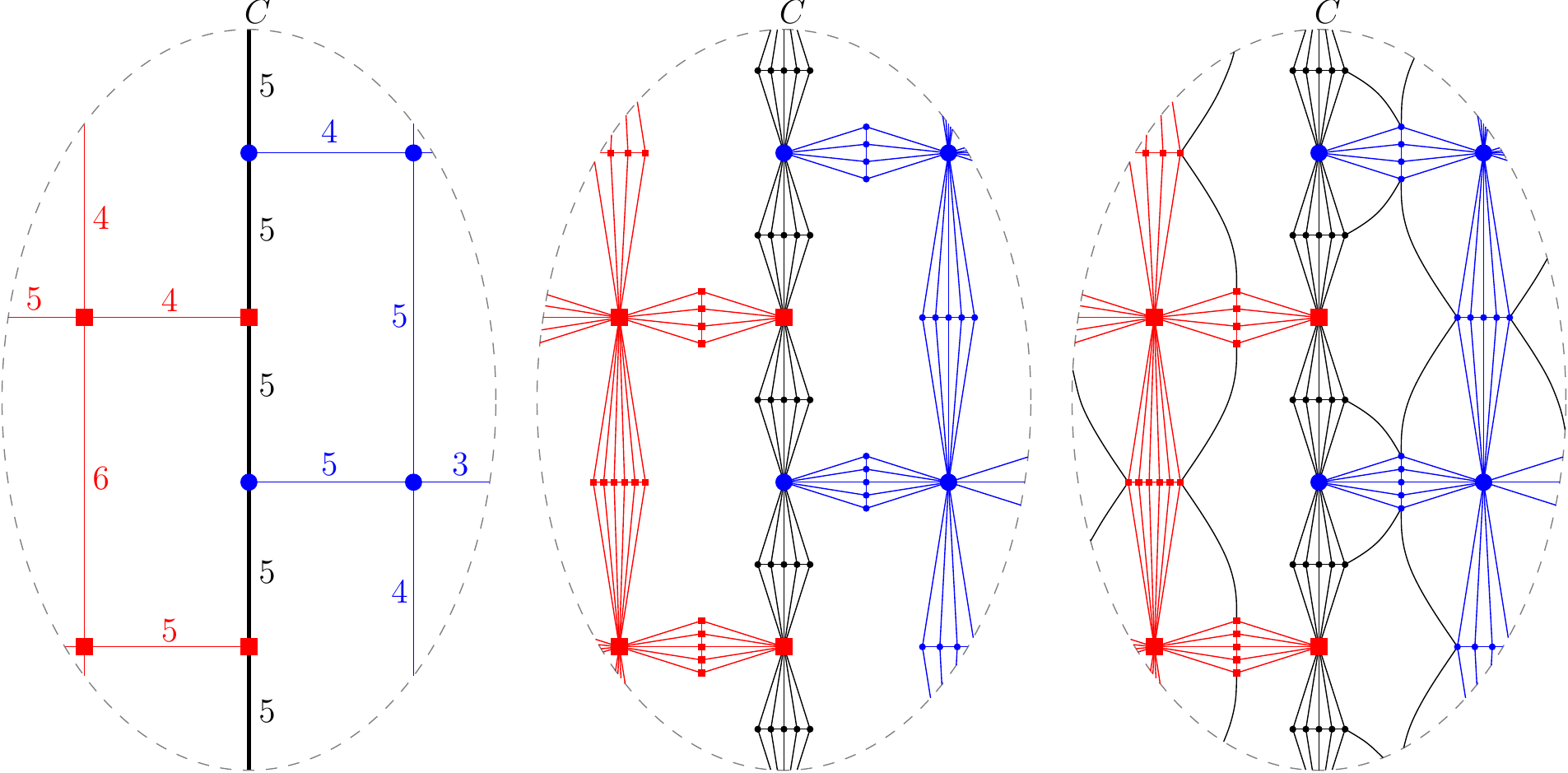}
	\caption{Example showing the transformation from $G'$ (left) to $G''$ (center) and to $\tilde G$ (right).
		The numbers on the left indicate the weights of the edges.}
	\label{fig:replacement}
\end{figure}

A \DEF{rotation system} in a graph $G$ is a list $\pi=(\pi_v)_{v\in V(G)}$,
where each $\pi_v$ is a cyclic permutation of the edges incident to vertex $v$.
A drawing of a graph $G$ with rotation system $\pi$ is a drawing where the 
the clockwise order of the edges incident to vertex $v$ in the drawing is given by the permutation $\pi_v$.
Pelsmajer et al.~\cite{PSS} have recently studied the crossing number with rotation systems.
Using a reduction from {\sc Linear Arrangement} they show that computing the crossing number
with rotation system is NP-hard. 
(As noted in~\cite{PSS}, we can keep working with weighted edges when studying this crossing number.)

Since the graphs $R$ and $B$ used in our construction have a unique anchored embedding, 
we already know a priori the rotation system in any drawing of $G$ whose restriction to $R$ or to $B$ is an embedding.
Thus Corollary~\ref{co:hard2} implies the following.

\begin{corollary}
\label{co:hard3} 
	The following problem is NP-hard: 
	given an anchored graph $G$ and a rotation system $\pi$ for $G$, 
	compute the minimum of $\CRA(\D)$
	over the anchored drawings $\D$ of $G$ with rotation system $\pi$.
	The problem remains hard when the graph $G$ is cubic.
\end{corollary}
\begin{proof}
	When the rotation system is fixed, we can replace 
	each vertex by a large hexagonal grid and attach
	the edges on the boundary of grid in the same order as in the rotation system.
	See Section 3 in~\cite{PSS} for the details. 
	(To keep the graph anchored, we just need to make several vertices in the boundary of
	the hexagonal grid anchors.)
\end{proof}

From Theorem~\ref{thm:hard10} we also obtain the following. 

\begin{corollary}
\label{co:hard4} 
	Computing the crossing number of a graph with fixed rotation system is NP-hard,
	even when the graph is planar and $3$-connected.
\end{corollary}
\begin{proof}
	Consider the graph $\tilde G$ used in the proof of Corollary~\ref{co:hard6}. 
	It is planar and $3$-connected. 
	For each non-anchor vertex of $(R\cup B)-V(C)$ we prescribe the rotation given by
	the unique combinatorial embedding of $R$ and of $B$ as anchored graphs.
	For each vertex of $C$, we prescribe the rotation system that
	forces $R\cup B$ to be drawn in the same side of $C$. Thus, the edge $rb$
	is not needed in the reduction.
\end{proof}

As shown in~\cite{PSS}, hardness of crossing number with rotation system implies
hardness of crossing number for cubic graphs by blowing up each vertex with a hexagonal
grid.
The described reduction yields a new, geometric 
proof of NP-completeness of the crossing number problem,
even when restricted to cubic graphs. Hardness of the crossing
number problem for cubic graphs was established by
Hlin\v{e}n\'y \cite{Hli}, who asked if one can prove this result 
by a reduction from an NP-complete geometric problem instead 
of the Optimal Linear Arrangement problem used in his proof.

A \DEF{rectilinear drawing} of a graph is a drawing where each edge is drawn using
a straight-line segment. The \DEF{rectilinear crossing number} of a graph $G$ is the minimum
number of crossings taken over all rectilinear drawings of $G$.
On the one hand, lower bounds for the crossing number of a graph $G$ 
are also lower bounds for the rectilinear crossing number of $G$.
On the other hand, the drawings used in Lemma~\ref{le:easydirection}
to show the upper bound on the crossing number are rectilinear drawings.
In fact, all the drawings we have described can be made straight-line drawings 
without increasing the number of crossings.
In particular, we can replace each vertex by a hexagonal grid and still keep a 
rectilinear drawing because in the drawing of Lemma~\ref{le:easydirection}
we know for most edges whether they are horizontal or vertical,
and for a few edges we have to decide whether they are diagonal or horizontal. 
From Corollaries~\ref{co:hard6} and~\ref{co:hard4} we
conclude the following:

\begin{corollary}
	The following problem is NP-hard: 
	given a simple planar, 3-connected graph $G$ and an additional edge $xy$,
	compute the \emph{rectilinear} crossing number of $G+xy$.
\end{corollary}

\begin{corollary}
	Computing the \emph{rectilinear} crossing number of cubic, 3-connected graphs is NP-hard.
\end{corollary}

It should be noted that in our reduction for near-planar graphs
we need high-degree vertices along the cycle $C$. 
We do not know the computational complexity of computing the crossing number of 
the graph $G+xy$ when $G$ is planar and has bounded degree, or even degree 4.
As discussed in the introduction, when $G$ has degree $3$ such crossing
number can be computed in polynomial time; see~\cite{Ri} and~\cite{CM}.

\section{1-planarity}
\label{S:1-planar}

In this section we show that it is NP-hard to decide whether a given near-planar graph 
is 1-planar. The approach is very similar to the one used to crossing number.
However, while crossing numbers are global, the concept of $1$-planarity is more local.
Thus, we first define a gadget and study its $1$-planar drawings.
Afterward we show that deciding $1$-planarity for anchored graphs is hard.
An example showing the eventual reduction for a small example is shown in Figure~\ref{fig:1-planar-example}.
Finally, we consider near-planar graphs. 

For any natural number $t$, a \DEF{$t$-path} is a path with $t$ edges.
A \DEF{$5$-thick edge} connecting vertices $u$ and $v$ is a set of five $2$-paths 
$u v_1 v, u v_2 v,\dots, u v_5 v$, where $v_1,\dots, v_5$ are distinct vertices of degree $2$.
Alternatively, such $5$-thick edge is a complete bipartite
graph $K_{2,5}$ with $u$ and $v$ defining one of its parts.
When a $5$-thick edge is drawn without any crossings, it defines 4 \DEF{inner faces}.
Those are the interior faces of the thick edge, each bounded by four edges.
(The unbounded face of the drawing is not an inner face.) In our drawings, we will represent
a $5$-thick edge with a filled in lens-like shape.

\begin{figure}
	\centering
	\includegraphics[width=.95\textwidth]{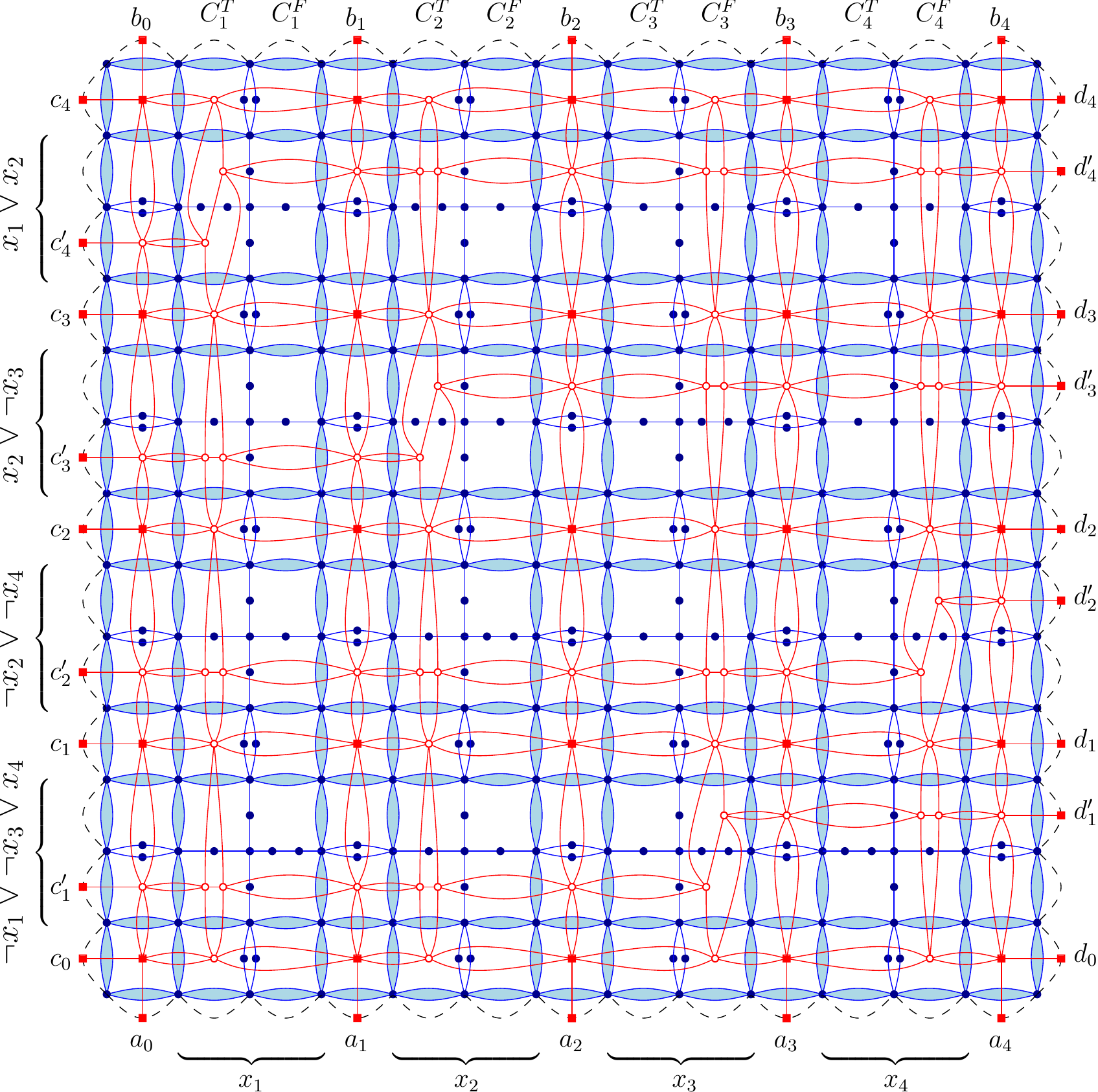}
	\caption{Example of the resulting reduction for anchored 1-planarity for the formula on 4 variables $x_1,x_2,x_3,x_4$
		and clauses $\neg x_1\vee \neg x_3\vee x_4$, $\neg x_2 \vee \neg x_4$,
		$x_2\vee \neg x_3$, $x_1\vee x_2$.
		The anchored 1-drawing in the figure corresponds to the boolean assignment $x_1=x_2=T$ and $x_3=x_4=F$.
		The dashed line around the figure indicates the boundary of the disk.}
		\label{fig:1-planar-example}
\end{figure}

\subsection{Basics and gadget}

We first provide some basic properties of $1$-planar graphs with $5$-thick edges.
If we have some $5$-thick edges in a graph $G$, we view each of them a single ($5$-thick) edge and
we no longer consider the vertices of degree $2$ inside these $5$-thick edges as vertices of $G$. 
After isolating as many edge-disjoint $5$-thick edges as possible, we cover all remaining vertices
of degree $2$ by maximal $t$-paths (of various lengths) whose internal vertices have degree $2$ in $G$.

To keep control over the drawings
it will be convenient to restrict our attention to $1$-drawings with the \emph{minimum number of crossings}.
For any drawing $\D$ of a graph $G$ and any subgraph $H$ of $G$, we use $\D_H$ for the
restriction of the drawing $\D$ to $H$.

\begin{lemma}
\label{le:5-thick}
	Let $G$ be a $1$-planar graph and let $H$ be its subgraph
	consisting of $5$-thick edges. 
	Let $\D$ be a $1$-drawing of $G$ with the minimum number of crossings.
	\begin{itemize}
	\item[\textup{(a)}] The restriction $\D_H$ of the drawing $\D$ to $H$ is an embedding.
	\item[\textup{(b)}] Let $z\in V(G)$ and $A\subset V(G)\setminus \{ z\}$. 
		Let $uv$ be a $5$-thick edge in $G$. 
		Suppose that there are $5$ edge-disjoint paths in $G- \{ u,v \}$ from $z$ 
		to vertices in $A$. 
		If in $\D_H$ none of the vertices in $A$ is in an inner face of the $5$-thick edge $uv$,
		then in $\D_H$ the vertex $z$ cannot be in an inner face of the $5$-thick edge $uv$.		
	\end{itemize}
	The same properties hold if $G$ is an anchored $1$-planar graph, 
	where no internal part of a $5$-thick edge of $H$ is an anchor, 
	and $\D$ is an anchored $1$-drawing of $G$ with the minimum number of crossings.
\end{lemma}

\begin{proof}
	We first show that in $\D$ no two parallel paths in the same $5$-thick edge intersect.
	Indeed, whenever any two such paths cross in a $1$-drawing, we can make a local
	change to get another $1$-drawing with strictly smaller total number of crossings.
	See Figure~\ref{fig:minimal} for the transformation.
	Observe that the same transformation can be used for anchored drawings as long as
	degree-2 vertices in $5$-thick edges are not anchors.
	
	We next argue that there are no crossings involving different $5$-thick edges.
	Consider a $5$-thick edge $uv$ and another $5$-thick edge $u'v'$.
	We already know that $\D_{uv}$ is an embedding. 
	The vertices $u'$ and $v'$ have to be in the same face of $\D_{uv}$:
	if they are in different faces, each of the $5$-edge disjoint paths forming $u'v'$ have
	to cross some of the edges of the face containing $u'$, which is not possible because such face
	has degree $4$.
	Let us consider the $2$-paths connecting $u'$ to $v'$.
	At least one of them, say $u' v_1 v'$, does not cross the boundary of the face $F$ of $\D_{uv}$
	that contains $u'$ and $v'$. If another $2$-path $u'v_2 v'$ crosses the boundary of $F$,
	then the face bounded by the $4$-cycle $u' v_1 v' v_2 u'$ either contains both $u$ and $v$ or none of them.
	It is easy to see that the $2$-path $u' v_2 v'$ is crossing twice a $2$-path $utv$ of
	the boundary of $F$	and we can unmake the crossings. 
	We conclude that no two $5$-edges of $H$ participate in a mutual crossing and thus $\D_H$ is an embedding.
	This completes the proof of item (a).
		
	To prove item (b) we use an argument similar to before.
	Assume for the sake of contradiction, that 
	the face $f_z$ in $\D_H$ that contains $z$ is an inner face of $uv$, and thus has degree $4$.
	Since we have $5$ edge-disjoint paths in $G-\{ u, v\}$ from $z$ to $A$, and $A$ has no vertex in $f_z$,
	each of those paths has to cross a distinct edge on the boundary of $f_z$, which is not possible.
\end{proof}

\begin{figure}
	\centering
	\includegraphics[width=.6\textwidth]{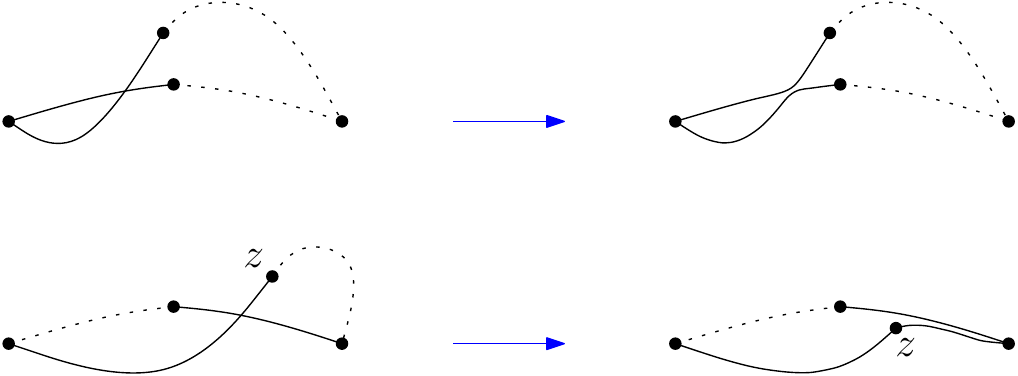}
	\caption{Transformation in Lemma~\ref{le:5-thick}(a). 
		Applying these transformations to a $1$-drawing produces another $1$-drawing with fewer crossings
		as long as $\deg(z)=2$ (and $z$ is not an anchor).}
		\label{fig:minimal}
\end{figure}

Our construction is based on a gadget consisting of a pair of graphs $X$ and $Y$.
We first describe the gadget and then analyze certain types of 1-drawings for $X\cup Y$.
 
Let $X$ be the \emph{embedded} graph constructed as follows (see Figure~\ref{fig:1-planar-notation}, left):
\begin{itemize}
	\item[(i)] The vertex set of $X$ is grid-like with vertices $u_{(\alpha,\beta)}$, $(\alpha,\beta)\in[4]\times[4]$.
	\item[(ii)] For each $(\alpha,\beta)\in[3]\times[4]$, $\beta\not= 2$
		we connect $u_{(\alpha,\beta)}$ and $u_{(\alpha+1,\beta)}$ with a $5$-thick edge.
	\item[(iii)] For each $(\alpha,\beta)\in[4]\times[3]$, $\alpha\not= 2$
		we connect $u_{(\alpha,\beta)}$ and $u_{(\alpha,\beta+1)}$ with a $5$-thick edge.
	\item[(iv)] We add $2$-thick edges between $u_{(0,2)}$ and $u_{(1,2)}$, between $u_{(3,2)}$ and $u_{(4,2)}$,
		between $u_{(2,0)}$ and $u_{(2,1)}$, and between $u_{(2,3)}$ and $u_{(2,4)}$.
	\item[(v)] We add $2$-path edges connecting $u_{(2,2)}$ to $u_{(1,2)}$, $u_{(2,3)}$, $u_{(3,2)}$ and to $u_{(2,1)}$.
	\item[(vi)] The embedding is the one obtained by assigning vertex $u_{(\alpha,\beta)}$ 
		to the point $(\alpha,\beta)\in[4]\times[4]$ in the Euclidean plane,
		and drawing the edges with almost straight curves. 
		Thus, the embedding makes $X$ look like a grid
	\item[(v)] We use $f_{(\alpha,\beta)}$ to denote the square-like face with $u_{(\alpha,\beta)}$ in its bottom left.
		The other faces, defined by edges within one single $5$-thick edge, are called \DEF{internal} faces.
		(Each $5$-thick edge has $4$ internal faces.)
\end{itemize}

\begin{figure}
	\centering
	\includegraphics[width=.8\textwidth]{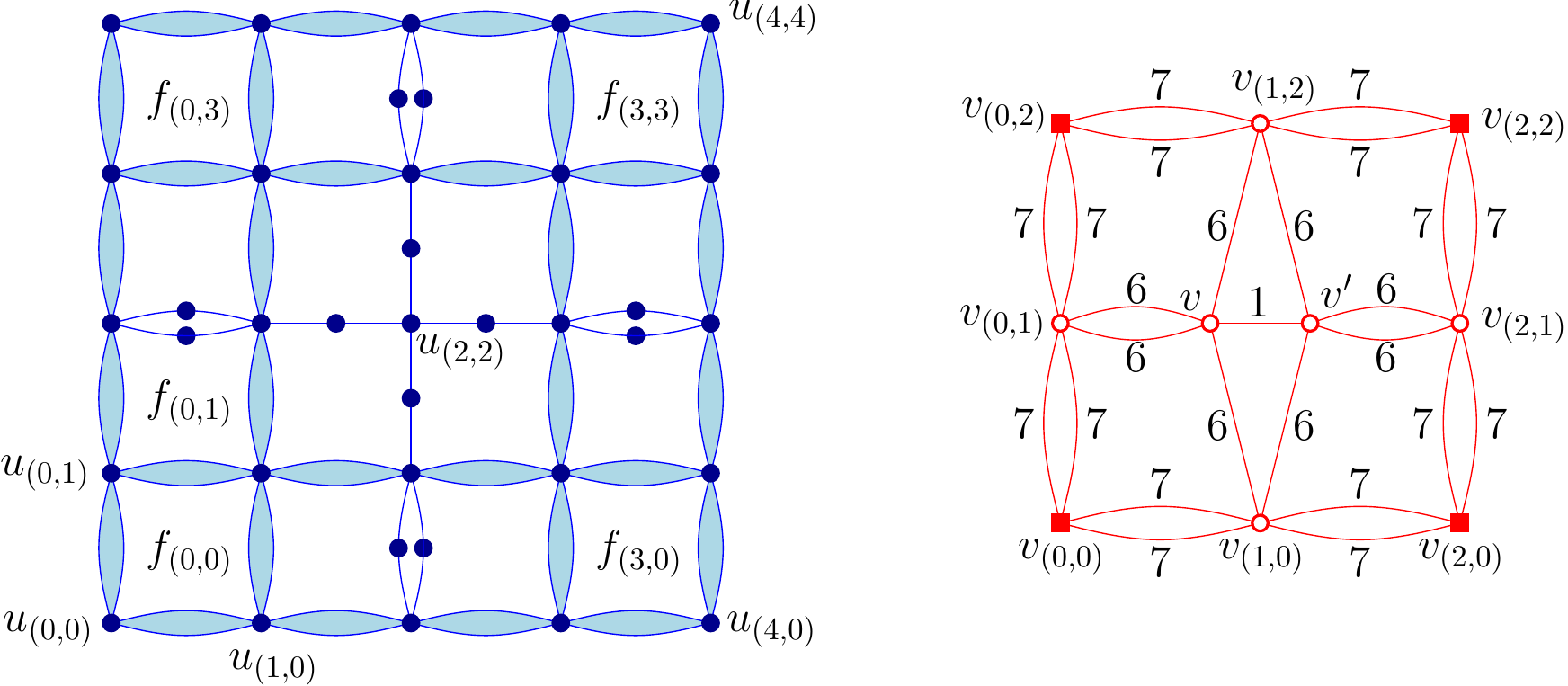}
	\caption{Left: the embedded graph $X$ with some of its vertices and faces labeled. The $5$-thick edges are shaded.
			Right: the graph $Y$ with its vertices labeled; each edge is actually a path whose length is annotated with the edge.}
	\label{fig:1-planar-notation}
\end{figure}

We call $X_\positive$ the variant of $X$ where the $2$-path connecting $u_{(1,2)}$ to $u_{(2,2)}$ 
is replaced with a $3$-path ($X_\positive$ is part of Figure~\ref{fig:1-planar-gadget-positive}, right).
We call $X_\negative$ the variant of $X$ where the $2$-path connecting $u_{(2,2)}$ to $u_{(3,2)}$ 
is replaced with a $3$-path. 

Let $Y$ be the graph, without a specified embedding, constructed as follows (see Figure~\ref{fig:1-planar-notation}, right):
\begin{itemize}
	\item[(i)] Start with a grid-like vertex set with vertices $v_{(\alpha,\beta)}$, $(\alpha,\beta)\in[2]\times[2]$,
		and add two disjoint $7$-paths between $v_{(\alpha,\beta)}$ and $v_{(\alpha',\beta')}$ if and only if 
		$|\alpha-\alpha'|+|\beta-\beta'|=1$.
	\item[(ii)] Remove $v_{(2,2)}$ and the $7$-paths incident to it.
	\item[(iii)] Add two new vertices $v$ and $v'$ and the edge $vv'$. 
		Next add two $6$-paths connecting $v$ to $v_{(0,1)}$ and two $6$-paths connecting $v'$ to $v_{(2,1)}$.
		Finally add $6$-paths connecting $v$ and $v'$ to $v_{(1,2)}$ and $v_{(1,0)}$.
\end{itemize}

Let $Z\in \{X, X_\positive,X_\negative\}$. 
A drawing of $Z\cup Y$ is \DEF{compliant} if it extends the embedding of $Z$,
no part of $Y$ is drawn in the outer face of $Z$,
vertex $v_{(0,0)}$ is in the face $f_{(0,0)}$ of $Z$, 
vertex $v_{(0,2)}$ is in the face $f_{(0,3)}$ of $Z$,
vertex $v_{(2,2)}$ is in the face $f_{(3,3)}$ of $Z$, and 
vertex $v_{(2,0)}$ is in the face $f_{(3,0)}$ of $Z$.
Let us call the vertices $v_{(0,0)}$, $v_{(0,2)}$, $v_{(2,2)}$, $v_{(2,0)}$ the \DEF{corners} of $Y$.
Thus a compliant drawing has a fixed position for the corners.
In our figures the corners are drawn with squares, 
while the other vertices of $Y$ are drawn with empty circles.
In the following we will only consider compliant $1$-drawings of $Z\cup Y$.
See Figure~\ref{fig:1-planar-gadget-positive} for some examples of compliant $1$-drawings.

For each non-corner vertex of $Y$ there are at least 5 edge-disjoint paths connecting it
to the corners. Since the position of the corners is fixed in a compliant $1$-drawing,
Lemma~\ref{le:5-thick}(b) implies that no non-corner vertex of $Y$ can be in an inner face of a $5$-thick edge of $Z$. 
Thus, in a compliant $1$-drawing of $Z\cup Y$ each vertex of $Y$ is in some face $f_{(\alpha,\beta)}$.

\begin{lemma}
\label{le:1-planar-gadget-positive}
	For any $Z\in \{X, X_\positive,X_\negative\}$,
	there is a compliant $1$-drawing of $Z\cup Y$ 
	satisfying any combination of a property stated in \textup(a\textup) and a property stated in \textup(b\textup) below:
	\begin{itemize}
		\item[\textup(a\textup)] $v_{(0,1)}$ is in the face $f_{(0,1)}$ and $v_{(2,1)}$ is in the face $f_{(3,1)}$,
			or $v_{(0,1)}$ is in the face $f_{(0,2)}$ and $v_{(2,1)}$ is in the face $f_{(3,2)}$.
		\item[\textup(b\textup)] $v_{(1,0)}$ is in the face $f_{(1,0)}$ and $v_{(1,2)}$ is in the face $f_{(1,3)}$,
			or $v_{(1,0)}$ is in the face $f_{(2,0)}$ and $v_{(1,2)}$ is in the face $f_{(2,3)}$.
	\end{itemize}
\end{lemma}
\begin{proof}
	See the left side of Figure~\ref{fig:1-planar-gadget-positive} for the drawing of $X\cup Y$.
	The remaining drawings of $X\cup Y$ are obtained by vertical and/or horizontal mirror symmetries.
	For $Z=X_\positive$ or $Z=X_\negative$ just subdivide once an edge incident to $u_{(2,2)}$.
\end{proof}

\begin{figure}
	\centering
	\includegraphics[width=.8\textwidth]{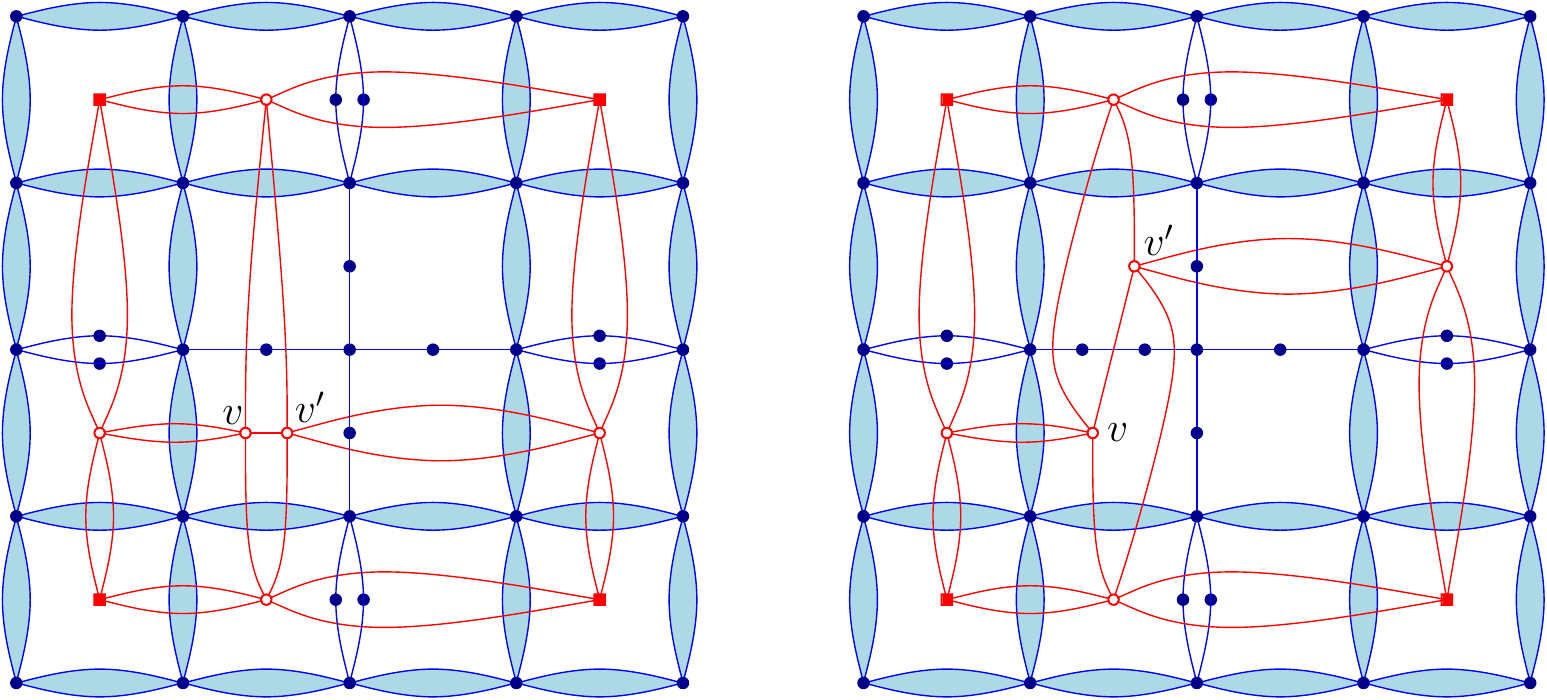}
	\caption{Left: A compliant drawing of $X\cup Y$ with the properties of Lemma~\ref{le:1-planar-gadget-positive}.
			Right: A compliant drawing of $X_\positive\cup Y$ with the properties of Lemma~\ref{le:1-planar-gadget-positive2}.}
	\label{fig:1-planar-gadget-positive}
\end{figure}

\begin{lemma}
\label{le:1-planar-gadget-positive2}
	There is a compliant $1$-drawing of $X_\positive\cup Y$ with the following two properties:
	\begin{itemize}
		\item $v_{(0,1)}$ is in the face $f_{(0,1)}$ and $v_{(2,1)}$ is in the face $f_{(3,2)}$;
		\item $v_{(1,0)}$ is in the face $f_{(1,0)}$ and $v_{(1,2)}$ is in the face $f_{(1,3)}$.
	\end{itemize}
	There is a compliant $1$-drawing of $X_\negative\cup Y$ with the following two properties:
	\begin{itemize}
		\item $v_{(0,1)}$ is in the face $f_{(0,1)}$ and $v_{(2,1)}$ is in the face $f_{(3,2)}$;
		\item $v_{(1,0)}$ is in the face $f_{(2,0)}$ and $v_{(1,2)}$ is in the face $f_{(2,3)}$.
	\end{itemize}
\end{lemma}
\begin{proof}
	See the right side of Figure~\ref{fig:1-planar-gadget-positive} for the drawing of $X_\positive\cup Y$.
	For $X_\negative\cup Y$ apply a vertical mirror symmetry.
\end{proof}

\begin{lemma}
\label{le:1-planar-gadget-pattern}
	For any $Z\in \{X, X_\positive,X_\negative\}$,
	any compliant $1$-drawing of $Z\cup Y$ has the following properties:
	\begin{itemize}
		\item[\textup{(a)}] vertex $v_{(0,1)}$ is in the face $f_{(0,1)}$ or $f_{(0,2)}$ of $Z$;
		\item[\textup{(b)}] vertex $v_{(2,1)}$ is in the face $f_{(3,1)}$ or $f_{(3,2)}$ of $Z$;
		\item[\textup{(c)}] vertex $v_{(1,0)}$ is in the face $f_{(1,0)}$ or $f_{(2,0)}$ of $Z$;
		\item[\textup{(d)}] vertex $v_{(1,2)}$ is in the face $f_{(1,3)}$ or $f_{(2,3)}$ of $Z$;
		\item[\textup{(e)}] if $v_{(1,0)}$ is in the face $f_{(1,0)}$ of $Z$,  
			then $v_{(1,2)}$ is in the face $f_{(1,3)}$ of $Z$;
		\item[\textup{(f)}] if $v_{(1,0)}$ is in the face $f_{(2,0)}$ of $Z$,  
			then $v_{(1,2)}$ is in the face $f_{(2,3)}$ of $Z$.
	\end{itemize}
\end{lemma}
\begin{proof}
	Consider any compliant $1$-drawing of $Z\cup Y$.
	The vertex $v_{(0,1)}$ has a path of length $7$ to $v_{(0,0)}$ and a path of length $7$ to $v_{(0,2)}$.
	This implies that $v_{(0,1)}$ must be in the face $f_{(0,1)}$ or $f_{(0,2)}$ of $Z$
	and proves item (a). Items (b)--(d) follow by rotational symmetry.
	
	To prove item (e), we note that there is a path of length $12$ connecting $v_{(1,0)}$ to $v_{(1,2)}$.
	By item (d), $v_{(1,2)}$ is in the face $f_{(1,3)}$ or $f_{(2,3)}$.
	However, any path from $f_{(1,0)}$ to $f_{(2,3)}$ must cross $13$ edges of $Z$.
	Thus $v_{(1,2)}$ must be in face $f_{(1,3)}$, and item (e) is proved. Item (f) follows from (e)
	by a vertical mirror symmetry. 
\end{proof}

\begin{lemma}	
\label{le:1-planar-gadget-impossible}
	If a compliant $1$-drawing of $X_\negative\cup Y$ has vertex $v_{(0,1)}$ in the face $f_{(0,1)}$ of $X_\negative$
	and $v_{(1,0)}$ is in the face $f_{(1,0)}$ of $X_\negative$,
	then $v_{(2,1)}$ must be in face $f_{(3,1)}$ of $X_\negative$.

	If a compliant $1$-drawing of $X_\positive\cup Y$ has vertex $v_{(0,1)}$ in the face $f_{(0,1)}$ of $X_\positive$
	and $v_{(1,0)}$ is in the face $f_{(2,0)}$ of $X_\positive$,
	then $v_{(2,1)}$ must be in face $f_{(3,1)}$ of $X_\positive$.

	If a compliant $1$-drawing of $X\cup Y$ has vertex $v_{(0,1)}$ in the face $f_{(0,1)}$ of $X$,
	then $v_{(2,1)}$ must be in face $f_{(3,1)}$ of $X$.
\end{lemma}
\begin{proof}
	We first consider the case for $X_\negative\cup Y$.	
	Consider any compliant $1$-drawing of $X_\negative\cup Y$ with the properties stated.
	Because of Lemma~\ref{le:1-planar-gadget-pattern}(e),
	vertex $v_{(1,2)}$ must be in the face $f_{(1,3)}$.
	The vertices $v$ and $v'$ have paths of length $6$ connecting to $v_{(1,0)}$ and $v_{(1,2)}$.
	This implies that $v$ and $v'$ can only be in faces $f_{(1,1)}$ or $f_{(1,2)}$, as any
	other face $f_{(\alpha,\beta)}$ requires at least $7$ crossings to connect to $v_{(1,0)}$ or $v_{(1,2)}$.
	We distinguish three cases:
	\begin{itemize}
	\item If the vertex $v'$ is in the face $f_{(1,1)}$ (Figure~\ref{fig:1-planar-gadget-impossible} left), 
		any path connecting $v'$ to the face $f_{(3,2)}$
		requires $7$ crossings. Since there is a path of length $6$ connecting $v'$ and $v_{(2,1)}$,
		the vertex $v_{(2,1)}$ cannot be in face $f_{(3,2)}$.
		Thus by Lemma~\ref{le:1-planar-gadget-pattern}(b) $v_{(2,1)}$ must be in the face $f_{(3,1)}$.
	\item If the vertex $v'$ is in the face $f_{(1,2)}$ and $v$ is in the face $f_{(1,1)}$,
		we argue as follows; see Figure~\ref{fig:1-planar-gadget-impossible}, right. 
		The $6$-path connecting $v$ to $v_{(1,2)}$, the $6$-path connecting $v'$ to $v_{(1,0)}$,
		and the edge $vv'$ must all cross the $2$-path connecting $u_{(1,2)}$ to $u_{(2,2)}$. Thus, no such
		compliant $1$-drawing with $v$ in $f_{(1,1)}$ and $v'$ in $f_{(1,2)}$ is possible.
	\item If the vertices $v'$ and $v$ are both in the face $f_{(1,2)}$, we note that
		the two $6$-paths connecting $v$ to $v_{(0,1)}$ and the $6$-path connecting $v$ to $v_{(1,0)}$
		must cross the $2$-path connecting $u_{(1,2)}$ to $u_{(2,2)}$. Thus, no such
		compliant $1$-drawing with $v'$ and $v$ in $f_{(1,2)}$ is possible.
	\end{itemize}
	This finishes the proof for $X_\negative\cup Y$.
	The claim for $X_\positive\cup Y$ follows from the claim for $X_\negative\cup Y$ by a vertical and horizontal mirror symmetry.
	For the case of $X\cup Y$, we note that $v_{(1,0)}$ must be in the face $f_{(1,0)}$ or $f_{(2,0)}$ by
	Lemma~\ref{le:1-planar-gadget-pattern}(c).
    If $v_{(1,0)}$ is in the face $f_{(1,0)}$ the claim follows from the case $X_\negative\cup Y$,
	and if $v_{(1,0)}$ is in the face $f_{(1,0)}$ then it follows from the case $X_\negative\cup Y$.
\end{proof}

\begin{figure}
	\centering
	\includegraphics[width=.8\textwidth]{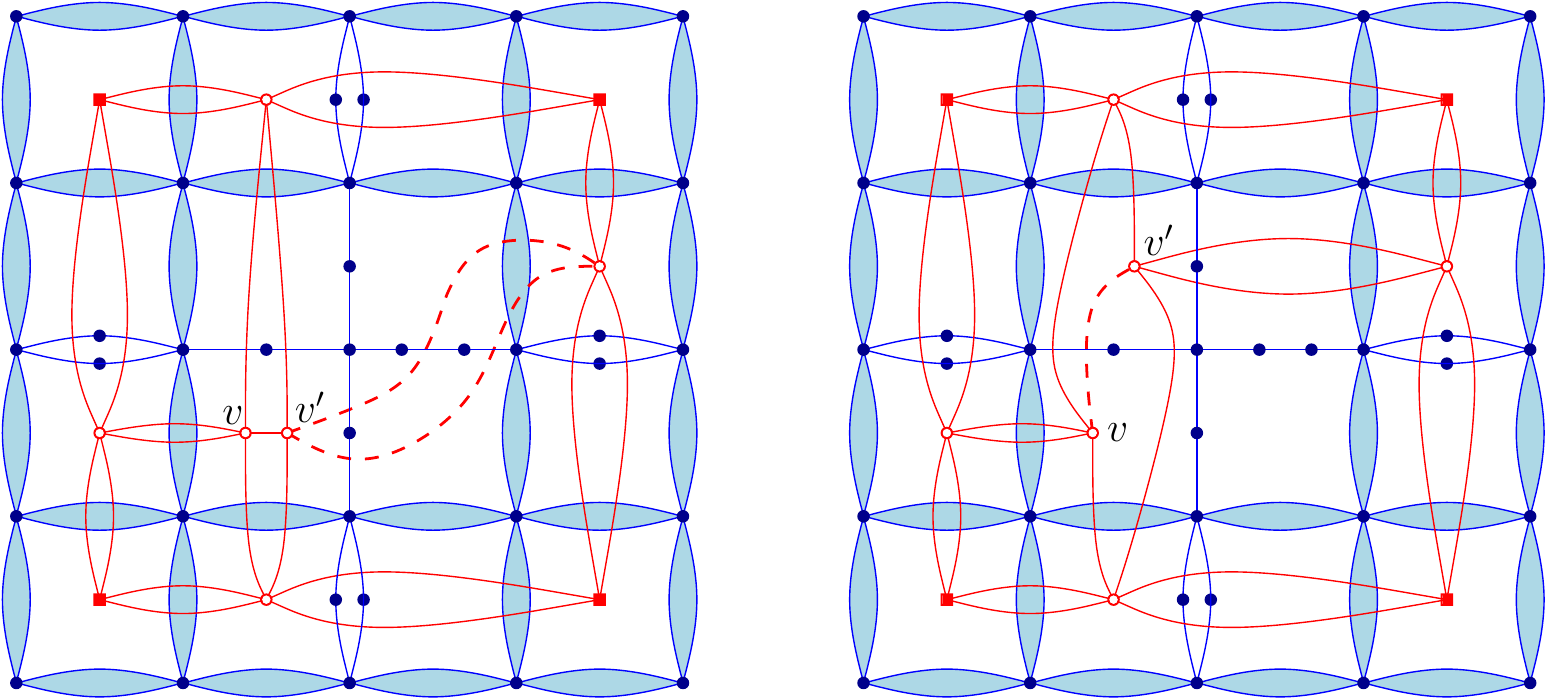}
	\caption{Analysis of the gadget $X_\negative\cup Y$ treated in Lemma~\ref{le:1-planar-gadget-impossible}.}
	\label{fig:1-planar-gadget-impossible}
\end{figure}

\subsection{1-planarity of anchored graphs}
\label{sec:anchored-1-planar}
In this section we prove the following result.

\begin{theorem}
\label{thm:1-planar-hard1} 
	Deciding if a given anchored graph is an anchored $1$-planar graph is NP-hard, 
	even if the input graph is decomposed
	into two vertex-disjoint planar anchored subgraphs
	(and the decomposition is part of the input).
\end{theorem}

We use a reduction from SAT and a grid-like construction similar to the one used for anchored crossing numbers in Section~\ref{sec:hard1}.
Like before, we consider an instance $I$ for SAT with variables $x_1,\dots, x_n$ and clauses $C_1,\dots,C_m$,
and construct two graphs $B=B(I)$ and $R=R(I)$.

The blue graph $B=B(I)$ is constructed as follows (see Figure~\ref{fig:1-planar-example}):
\begin{itemize}
\item[(i)] For each variable $x_i$ and each clause $C_j$, 
	we make an embedded graph $B^{(i,j)}$ such that 
	$B^{(i,j)}$ is a copy of $X_\positive$ if the literal $x_i$ appears in $C_j$, 
	a copy of $X_\negative$ if the literal $\neg x_i$ appears in $C_j$,
	and a copy of $X$ otherwise. 
	We use $u^{(i,j)}_{(\alpha,\beta)}$ and $f^{(i,j)}_{(\alpha,\beta)}$ for the vertex $u_{(\alpha,\beta)}$ 
	and face $f_{(\alpha,\beta)}$ of $B^{(i,j)}$, respectively.
\item[(ii)] We identify parts of the graphs $B^{(i,j)}$ as follows. 
	For each clause $C_j$, $j< m$, and each $\alpha\in [3]$, we identify
	the faces and $5$-thick edges of $f^{(i,j)}_{(\alpha,3)}$ and of $f^{(i,j+1)}_{(\alpha,0)}$ ($i=1,\dots,n$).
	For each variable $x_i$, $i< n$, and each $\beta\in [3]$, we identify
	the faces and $5$-thick edges of $f^{(i,j)}_{(3,\beta)}$ and of $f^{(i+1,j)}_{(0,\beta)}$ ($j=1,\dots,m$).
\item[(iii)] We disregard the embedding of the graph, and consider it as an abstract graph.
\item[(iv)] The anchors of $B$ are the vertices $u^{(1,j)}_{(0,\beta)}$,
	$u^{(n,j)}_{(4,\beta)}$, $u^{(i,1)}_{(\alpha,0)}$, and $u^{(i,m)}_{(\alpha,4)}$, where
	$\alpha,\beta\in [4]$, $i=1,\dots,n$ and $j=1,\dots, m$.
\end{itemize}

For each variable $x_i$, we define two \DEF{columns}.
The column $C_i^T$ is formed by the faces $f^{(i,j)}_{(1,\beta)}$, where $j=1,\dots ,m$ and $\beta\in [2]$.
The column $C_i^F$ is formed by the faces $f^{(i,j)}_{(2,\beta)}$, where $j=1,\dots ,m$ and $\beta\in [2]$.
In Figure~\ref{fig:1-planar-example}, the columns are annotated on the top.

The red graph $R=R(I)$ is constructed as follows (see Figure~\ref{fig:1-planar-example}):
\begin{itemize}
\item[(i)] For each variable $x_i$ and each clause $C_j$,
	we make a graph $R^{(i,j)}$ such that 
	$R^{(i,j)}$ is a copy of $Y$. 
	We use $v^{(i,j)}_{(\alpha,\beta)}$ for the vertex $v_{(\alpha,\beta)}$ of $R^{(i,j)}$.
\item[(ii)] We identify parts of the graphs $R^{(i,j)}$ as follows. 
	For each clause $C_j$, $j< m$, and each $\alpha\in [2]$, we identify
	the vertex $v^{(i,j)}_{(\alpha,2)}$ with $v^{(i,j+1)}_{(\alpha,0)}$ ($i=1,\dots,n$).
	For each variable $x_i$, $i< n$, and each $\beta\in [2]$, we identify
	the vertex $v^{(i,j)}_{(3,\beta)}$ and of $v^{(i+1,j)}_{(0,\beta)}$ ($j=1,\dots,m$).
	We also identify, pairwise, the two $7$-paths that connect identified vertices.
\item[(iii)] For each $i\in [n]$ we create two vertices, $a_i$ and $b_i$. 
	For each variable $x_i$, vertex $a_i$ is connected to $v^{(i,1)}_{(2,0)}$ with a $5$-path
	and vertex $b_i$ is connected to $v^{(i,m)}_{(2,2)}$ with a $5$-path.
	Vertex $a_0$ is connected to $v^{(1,1)}_{(0,0)}$ with a $5$-path
	and vertex $b_0$ is connected to $v^{(1,m)}_{(0,2)}$ with a $5$-path.
\item[(iv)] For each $j=1,\dots, m$ we create four vertices, $c_j,c'_j,d_j,d'_j$. 
	For each clause $C_j$, vertex $c_j$ is connected to $v^{(1,j)}_{(0,2)}$ with a $5$-path,
	vertex $c'_j$ is connected to $v^{(1,j)}_{(0,1)}$ with a $5$-path,
	vertex $d_j$ is connected to $v^{(n,j)}_{(2,2)}$ with a $5$-path,
	and vertex $d'_j$ is connected to $v^{(n,j)}_{(2,1)}$ with a $5$-path.
	We also create two vertices $c_0$ and $d_0$.
	Vertex $c_0$ is connected to $v^{(1,1)}_{(0,0)}$ with a $5$-path
	and vertex $d_0$ is connected to $v^{(n,1)}_{(2,0)}$ with a $5$-path.	
	The vertices created in this step are the anchors of $R$.
\end{itemize}

Let $G = G(I)$ be the anchored graph obtained by taking the union of the red graph $R$ and the blue
graph $B$. The clockwise circular ordering of the anchors along the boundary of the disk is defined
by the following properties:
\begin{itemize}
	\item For each variable $x_i$, we have the subsequences of anchors
		$a_{i}, u^{(i,1)}_{(3,0)}u^{(i,1)}_{(2,0)},u^{(i,1)}_{(1,0)} , a_{i-1}$
		and $b_{i-1}, u^{(i,m)}_{(1,4)},u^{(i,m)}_{(2,4)},u^{(i,m)}_{(3,4)} , b_{i}$.
	\item For each clause $C_j$, we have the subsequences of anchors
		$c_{j-1}, u^{(1,j)}_{(0,1)},c'_{j},u^{(1,j)}_{(0,2)},u^{(1,j)}_{(0,3)}, c_{j}$
		and $d_{j}, u^{(n,j)}_{(4,3)},d'_j, u^{(n,j)}_{(4,2)},u^{(n,j)}_{(4,1)}, d_{j-1}$.
	\item We have the subsequence $a_0,u^{(1,1)}_{(0,0)},c_0$, the subsequence $c_m,u^{(1,m)}_{(0,4)},b_0$, 
		the subsequence $b_n,u^{(n,m)}_{(4,4)},d_m$, and the subsequence $d_0,u^{(n,1)}_{(4,0)},a_n$.
\end{itemize}

This concludes the description of the graph $G(I)$.

\begin{lemma}
\label{le:1-planar-easy}
	If the instance $I$ is satisfiable, then $G(I)$ has an anchored $1$-drawing.
\end{lemma} 
\begin{proof}
	We draw the graph $B$ without crossings. The corresponding embedding is unique,
	up to permutations of parallel $2$-paths.
	Thus, the embedding of $B$ corresponds to the one used during its construction,
	and we can talk about the subgraphs $B^{(i,j)}$ and their faces $f^{(i,j)}_{(\alpha,\beta)}$.
	
	Let $q_i\in \{T,F \}$ be an assignment for each variable $x_i$ that satisfies all clauses.  
	For each clause $C_j$, let $x_{t(j)}$ be the first variable
	whose value $q_{t(j)}$ makes the clause $C_j$ true. 
	For each variable $x_i$ and each clause $C_j$, 
	the graph $R^{(i,j)}$ is drawn according to the following cases:
	\begin{description}
	\item[Case $i=t(j)$.] We use the compliant $1$-drawing of $B^{(i,j)}\cup R^{(i,j)}$ 
		given in Lemma~\ref{le:1-planar-gadget-positive2}.
	\item[Case $i<t(j)$.] We use a drawing of $B^{(i,j)}\cup R^{(i,j)}$ 
		given by Lemma~\ref{le:1-planar-gadget-positive}
		with $v^{(i,j)}_{(0,1)}$ in $f^{(i,j)}_{(0,1)}$,
		$v^{(i,j)}_{(2,1)}$ in $f^{(i,j)}_{(3,1)}$, and with both
		$v^{(i,j)}_{(1,0)}$ and $v^{(i,j)}_{(1,2)}$ in $C_i^{q_i}$.
	\item[Case $i>t(j)$.] We use a drawing of $B^{(i,j)}\cup R^{(i,j)}$ 
		given by Lemma~\ref{le:1-planar-gadget-positive}
		with $v^{(i,j)}_{(0,1)}$ in $f^{(i,j)}_{(0,2)}$,
		$v^{(i,j)}_{(2,1)}$ in $f^{(i,j)}_{(3,2)}$, and with both
		$v^{(i,j)}_{(1,0)}$ and $v^{(i,j)}_{(1,2)}$ in $C_i^{q_i}$.
	\end{description}	
	It is easy to check that, whenever a vertex or a path
	appears in more than one subgraph $R^{(i,j)}$, 
	the drawing we have described in both subgraphs is the same. 
	Therefore, we have described a drawing for $R$ minus its set of anchors.
	The drawing of the $5$-paths incident to the anchors is straightforward.
\end{proof}

\begin{lemma}
\label{le:1-planar-hard}
	If $G=G(I)$ is an anchored $1$-planar graph, then $I$ is satisfiable.
\end{lemma} 
\begin{proof}
	Assume that $G$ is an anchored $1$-planar graph and consider
	an anchored $1$-drawing $\D$ of $G$ with the minimum number of crossings.
	Let $H$ be the subgraph of $G$ consisting of $5$-thick edges.
	By Lemma~\ref{le:5-thick}(a), the restriction $\D_H$ is an embedding.
	The embedding of $\D_H$ is unique, up to permutation of parallel $2$-paths,
	because $H$ is essentially a subdivision of a grid.
	We can further argue that $\D_B$ is an embedding. 
	Indeed, if any of the $2$- or $3$-paths contained in $B^{(i,j)}$ would
	participate in some crossing in $\D_B$, 
	they can be redrawn to obtain another $1$-drawing with fewer crossings. 

	The embedding $\D_B$ is unique,
	up to permutations of parallel $2$-paths.
	Thus, $\D_B$ corresponds to the embedding used during the previous discussion,
	and we can talk about the subgraphs $B^{(i,j)}$ and their faces $f^{(i,j)}_{(\alpha,\beta)}$.
	By Lemma~\ref{le:5-thick}(b), the vertices of $R$ cannot be in any inner face of $\D_B$.
	Thus, each vertex of $R$ is in some face $f^{(i,j)}_{(\alpha,\beta)}$. 
	
	For each variable $x_i$ and each clause $C_j$, the anchors $a_{i}$, $b_{i}$, $c_j$, $d_j$ 
	force that the vertex $v^{(i,j)}_{(2,2)}$ is in the face $f^{(i,j)}_{(3,3)}$. 
	Indeed, if $v^{(i,j)}_{(2,2)}$ would be in any other face, at least one of the paths connecting it
	to the anchors cannot be drawn with its vertices being in non-inner faces.
	Similar statements hold for each corner of each $R^{(i,j)}$ and, as a consequence,
	the restriction of $\D$ to any $B^{(i,j)}\cup R^{(i,j)}$ is a compliant $1$-drawing.
	Furthermore, for each clause $C_j$ the anchors $c'_j$ and $d'_j$ force that
	$v^{(1,j)}_{(0,1)}$ is in $f^{(1,j)}_{(0,1)}$ and $v^{(n,j)}_{(2,1)}$ is in $f^{(n,j)}_{(3,2)}$.
	
	Consider any variable $x_i$.
	Because of Lemma~\ref{le:1-planar-gadget-pattern}(e) and (f), the vertices $f^{(i,j)}_{(1,0)}$, $j=1,\dots,m$,
	are all in the column $C_i^T$ or the column $C_i^F$. In the former case, we define $q_i=T$, and
	in the latter we define $q_i=F$.
	
	We next argue that the assignment $\{ x_i=q_i \}_i$ satisfies all clauses $C_1,\dots,C_m$,
	which implies that $I$ is satisfiable.
	To see this assume, for the sake of contradiction, that some $C_j$ is not satisfied.
	This means that, whenever $x_i=T$, $B^{(i,j)}$ is a copy of $X$ or $X_\negative$,
	and whenever $x_i=F$, $B^{(i,j)}$ is a copy of $X$ or $X_\positive$.
	Since $v^{(1,j)}_{(0,1)}$ is in the face $f^{(1,j)}_{(0,1)}$,
	consecutive applications of Lemma~\ref{le:1-planar-gadget-impossible} imply
	that $v^{(i,j)}_{(2,1)}$ is in the face $f^{(i,j)}_{(3,1)}$ for each $i=1,\dots,n$,
	and therefore $v^{(n,j)}_{(2,1)}$ is in the face $v^{(n,j)}_{(3,1)}$. This contradicts
	the fact that in $\D$ the vertex $v^{(n,j)}_{(2,1)}$ is in the face $f^{(n,j)}_{(3,2)}$, as discussed before.
\end{proof}

Theorem~\ref{thm:1-planar-hard1} follows from Lemmas~\ref{le:1-planar-easy} and~\ref{le:1-planar-hard}
because the construction of $G$ takes linear time and has maximum degree $20$.

\subsection{1-planarity of near-planar graphs}

\begin{theorem}
\label{thm:hard3} 
	Deciding whether a given near-planar graph is $1$-planar is NP-complete,
	even when the graph has bounded maximum degree, has bounded crossing number, and can
	be drawn in the plane so that one edge crosses two other edges, but every other edge participates in at most one crossing.
\end{theorem}
\begin{proof}
	The proof is very similar to the proof of Theorem~\ref{thm:hard10},
	but in this case it will be convenient 
	to build on the precise construction used in Section~\ref{sec:anchored-1-planar}.
	
	Let $G$ be the anchored graph constructed in Section~\ref{sec:anchored-1-planar}.
	Consider the graph $G'$ obtained from $G$ as follows.
	For every two consecutive anchors $a$ and $a'$ of $G$ in the cyclic ordering $\pi_G$,
	we introduce in $G'$ a $5$-thick edge connecting $a$ to $a'$. 
	The set of added edges defines a $5$-thick cycle, which we denote by $C$.
	Finally, we add a $9$-path between $u^{(1,1)}_{(1,1)}$ and $v^{(1,1)}_{(0,0)}$,
	and let $e$ denote an edge on this path.
	
	The graph $G'$ is a near-planar graph because $G'-e$ is planar: we can draw $C$ as a cycle,
	draw $B$ inside $C$ planarly, and draw $R$ outside $C$ planarly. 
	Adding the edge $e$ to this drawing shows that $\CR(G')\le 10$. All edges except $e$ in this drawing obey
	the $1$-planarity condition and it can be achieved that $e$ crosses only two other edges.
	
	The instance $I$ is satisfiable if and only if $G'$ is $1$-planar.
	Indeed, when $I$ is satisfiable, the drawing of $G$
	described in Lemma~\ref{le:1-planar-easy} can be extended to a drawing $G'$ because
	$e$ and $C$ can be added without using additional crossings.
	When $G'$ is $1$-planar, its $1$-drawing $\D$ with the minimum number of crossings
	has the property that $\D_C$ is an embedding because of Lemma~\ref{le:5-thick}(a).
	The edge $e$ forces that $G$ is drawn inside or outside $C$ because $e$ can
	only participate in one crossing. Therefore, the restriction of $\D$ to $G$
	is an anchored $1$-drawing of $G$, 
	which implies that $I$ is satisfiable by Lemma~\ref{le:1-planar-hard}.
	Since $G'$ can be constructed from $G$ in linear time and $G'$ has maximum
	degree 20, the result follows.
\end{proof}

\section*{Acknowledgments} We thank anonymous reviewers of~\cite{cm-socg-2010} for 
several helpful suggestions.

\bibliographystyle{abbrv}
\bibliography{biblio}
\end{document}